\documentclass[preprint,nonatbib]{elsarticle}
\usepackage{cite}
\usepackage[utf8]{inputenc}
\usepackage[misc]{ifsym}

\newif\ifbuildtikz
\buildtikzfalse

\usepackage{xpatch}
\usepackage{etoolbox}
\newcommand{\emptycmd}[1]{}
\usepackage[caption=false]{subfig}

\usepackage{booktabs}
\usepackage{multirow}
\usepackage[math]{cellspace}

\usepackage{hyphenat}

\usepackage[english]{babel}

\usepackage{xspace}

\usepackage{amsmath}
\usepackage{amssymb}
\usepackage{amsthm}

\usepackage[extdef]{delimset}

\newcommand{\floor}{\delim\lfloor\rfloor}
\newcommand{\etal}{et al.\xspace}

\usepackage{xstring}

\newcommand{\sig}[1]{\ensuremath{\sigma_{#1}}}

\newcommand{\cc}[1]{\textnormal{#1}\xspace}

\newcommand{\NP}{\cc{NP}}

\newcommand{\eps}{\varepsilon}
\newcommand{\N}{\ensuremath{\mathbb{N}}}
\newcommand{\R}{\ensuremath{\mathbb{R}}}

\newcommand{\polar}{\delimpair({[.];})}

\newcommand{\textfrac}[2]{\IfSubStr{#1}{+}{\brk*{#1}}{\IfSubStr{#1}{-}{\brk*{#1}}{#1}}/#2}

\usepackage[autostyle=true]{csquotes}

\usepackage[inline]{enumitem}

\setlist[enumerate,1]{label={(\roman*)}}
\newlist{theoremcases}{enumerate}{1}
\setlist[theoremcases,1]{label={(\roman*)}}
\newlist{proofcases}{enumerate}{3}
\setlist[proofcases,1]{label={(\roman*)}}
\setlist[proofcases,2]{label={(\arabic*)}}
\setlist[proofcases,3]{label={(\alph*)}}

\usepackage[%
    ,svgnames%
    ,hyperref%
    ]{xcolor}

\definecolor{fu-blue}{RGB}{0, 51, 102}
\definecolor{fu-green}{RGB}{178, 204, 51}
\definecolor{fu-red}{RGB}{204, 0, 0}
\definecolor{fu-orange}{RGB}{255, 153, 0}

\definecolor{fu-blue-90}{RGB}{22, 69, 116}
\definecolor{fu-blue-80}{RGB}{44, 86, 130}
\definecolor{fu-blue-70}{RGB}{66, 104, 144}
\definecolor{fu-blue-60}{RGB}{88, 122, 158}
\definecolor{fu-blue-50}{RGB}{109, 139, 172}
\definecolor{fu-blue-40}{RGB}{131, 157, 186}
\definecolor{fu-blue-30}{RGB}{153, 175, 200}
\definecolor{fu-blue-20}{RGB}{175, 192, 214}
\definecolor{fu-blue-10}{RGB}{197, 210, 228}

\definecolor{fu-black}{gray}{0}
\definecolor{fu-gray-90}{gray}{.1}
\definecolor{fu-gray-80}{gray}{.2}
\definecolor{fu-gray-70}{gray}{.3}
\definecolor{fu-gray-60}{gray}{.4}
\definecolor{fu-gray-50}{gray}{.5}
\definecolor{fu-gray-40}{gray}{.6}
\definecolor{fu-gray-30}{gray}{.7}
\definecolor{fu-gray-20}{gray}{.8}
\definecolor{fu-gray-10}{gray}{.9}
\definecolor{fu-gray-5}{gray}{.95}

\usepackage{textcomp}
\usepackage{siunitx}

\usepackage{relsize}

\usepackage{microtype}

\usepackage{listings}
\ifbuildtikz
\usepackage{tikz}
\usetikzlibrary{external}
\tikzset{external/only named}
\tikzset{external/system call={%
    lualatex \tikzexternalcheckshellescape -halt-on-error -interaction=nonstopmode -jobname "\image" "\texsource"%
}}
\let\tikzexternalizenext\tikzsetnextfilename
\let\tikzmaybeexternalizenext\tikzsetnextfilename
\else
\let\tikzexternalizenext\emptycmd
\let\tikzmaybeexternalizenext\emptycmd
\fi

\usepackage[%
    ,colorinlistoftodos%
    ,obeyFinal%
    ,color=fu-orange%
    ,bordercolor=fu-red%
    ,textsize=small%
    ]{todonotes}
\usepackage{soul}
\usepackage{soulutf8}

\makeatletter
\patchcmd{\SOUL@ulunderline}{\dimen@}{\SOUL@dimen}{}{}
\patchcmd{\SOUL@ulunderline}{\dimen@}{\SOUL@dimen}{}{}
\patchcmd{\SOUL@ulunderline}{\dimen@}{\SOUL@dimen}{}{}
\newdimen\SOUL@dimen
\makeatother

\sethlcolor{fu-orange}

\makeatletter
\xpretocmd{\todo}{\@bsphack}{}{}
\xapptocmd{\todo}{\@esphack}{}{}
\makeatother

\tikzset{%
    notestyle/.append style={%
        rounded corners=0pt,%
    }%
}

\ifbuildtikz
\usepackage{letltxmacro}
\LetLtxMacro{\oldmissingfigure}{\missingfigure}
\renewcommand{\missingfigure}[2][]{\tikzexternaldisable\oldmissingfigure[{#1}]{#2}\tikzexternalenable}
\LetLtxMacro{\oldtodo}{\todo}
\renewcommand{\todo}[2][]{\tikzexternaldisable\oldtodo[#1]{#2}\tikzexternalenable}

\usepackage{ifluatex}
\ifluatex
\usetikzlibrary{graphdrawing}
\usegdlibrary{trees}
\fi

\usetikzlibrary{calc}
\usetikzlibrary{fit}
\usetikzlibrary{shapes}
\usetikzlibrary{patterns}
\usetikzlibrary{intersections}
\usetikzlibrary{graphs}
\usetikzlibrary{spy}
\usetikzlibrary{angles}
\usetikzlibrary{quotes}

\tikzset{%
    every path/.append style={line join=round},%
}
\makeatletter
\tikzset{clip/.code={%
    \let\tikz@mode=\pgfutil@empty%
    \let\tikz@preactions=\pgfutil@empty%
    \let\tikz@postactions=\pgfutil@empty%
    \let\tikz@options=\pgfutil@empty%
    \tikz@addmode{\tikz@mode@cliptrue}%
  },
}
\makeatother

\tikzstyle{vertex}=[%
    draw,%
    circle,%
    minimum width=0.15cm,%
    fill,%
    outer sep=0,%
    inner sep=0,%
    ]
\tikzstyle{beacon}=[%
    vertex,%
    fill=fu-red,%
    ]
\tikzstyle{optional vertex}=[%
    vertex,%
    fu-gray-50,%
    ]

\tikzstyle{polygon}=[%
    draw,%
    ]
\tikzstyle{filled polygon}=[%
    polygon,%
    fill=fu-gray-5,%
    ]
\tikzstyle{optional polygon}=[%
    polygon,%
    dotted,%
    ]
\tikzstyle{polygon triangulation}=[%
    polygon,%
    gray,%
    ]
\tikzstyle{filled polygon triangulation}=[%
    polygon,%
    dotted,%
    ]
\tikzstyle{beacon boundary}=[%
    polygon triangulation,%
    very thick,%
    dashed,%
    fu-blue,%
    ]

\tikzstyle{polytope}=[%
    draw,%
    ]
\tikzstyle{polytope hidden}=[%
    polytope,%
    dashed,%
    ]
\tikzstyle{polytope triangulation}=[%
    polytope,%
    dotted,%
    ]

\tikzstyle{shared vertex}=[%
    vertex,%
    fill=fu-orange,%
    ]
\tikzstyle{shared edge}=[%
    black,%
    double=fu-orange,%
    double distance=1.2pt,%
    thin,%
    ]

\tikzstyle{attraction path}=[%
    draw,%
    very thick,%
    fu-red,%
    ]
\tikzstyle{virtual attraction path}=[%
    draw,%
    very thick,%
    fu-red,%
    dotted,%
    ]

\tikzstyle{primary coordinate axis}=[%
    draw,%
    ->,%
    gray,%
    dashed,%
    ]
\tikzstyle{secondary coordinate axis}=[%
    draw,%
    gray,%
    dotted,%
    ]

\tikzstyle{removed tetrahedra}=[%
    fu-red,%
    ]

\tikzstyle{invertedclip}=[%
    clip,%
    insert path={{%
        [reset cm]%
        (-16383.99999pt,-16383.99999pt)%
        rectangle%
        (16383.99999pt,16383.99999pt)%
    }}%
]
\fi

\newcommand{\tetrahedron}{%
    \tikz[x=1.5ex, y=1.5ex, baseline=-0.05ex, line width=0.4]{%
        \draw (0.45,1) -- (0,0) -- (.9,0) -- (0.45,1) -- (1.2,0.3) -- (.9,0);
        \path (0,0) -- (1.3,0);
        }
    }

\newcommand{\trapezoid}{%
    \tikz[x=1.5ex, y=1.5ex, baseline=-0.05ex, line width=0.4]{%
        \draw (0,0) -- (1.1,0) -- (0.8,1) -- (0.3,1) -- cycle;
        \path (0,0) -- (1.3,0);
        }
    }

\usepackage{hyperref}
\usepackage[capitalise]{cleveref}

\newdefinition{definition}{Definition}
\newdefinition{observation}[definition]{Observation}
\newdefinition{note}[definition]{Note}
\newtheorem{theorem}[definition]{Theorem}
\newtheorem{lemma}[definition]{Lemma}

\ifdef{\qedhere}{}{
\newif\ifqed
\newcommand\qedhere{\global\qedfalse\hfill\ensuremath{\blacksquare}}
\appto\proof{\global\qedtrue}
\preto\endproof{\ifqed\hfill\ensuremath{\blacksquare}\fi}
}

\crefname{problem}{Problem}{Problems}
\Crefname{problem}{Problem}{Problems}
\crefname{assumption}{Assumption}{Assumptions}
\Crefname{assumption}{Assumption}{Assumptions}
\crefname{openquestion}{Open question}{Open questions}
\Crefname{openquestion}{Open question}{Open questions}
\crefname{observation}{Observation}{Observations}
\Crefname{observation}{Observation}{Observations}
\crefname{hypothesis}{Hypothesis}{Hypotheses}
\Crefname{hypothesis}{Hypothesis}{Hypotheses}

\crefname{theoremcasesi}{case}{cases}
\Crefname{theoremcasesi}{Case}{Cases}

\crefname{proofcasesi}{case}{cases}
\Crefname{proofcasesi}{Case}{Cases}
\crefname{proofcasesii}{subcase}{subcases}
\Crefname{proofcasesii}{Subcase}{Subcases}
\crefname{proofcasesiii}{subsubcase}{subsubcases}
\Crefname{proofcasesiii}{Subsubcase}{Subsubcases}

\crefname{subsfigure}{subfigure}{subfigures}
\Crefname{subsfigure}{Subfigure}{Subfigures}

\makeatletter
\DeclareRobustCommand{\nosortcref}[1]{%
  \begingroup\@cref@sortfalse\cref{#1}\endgroup
}
\DeclareRobustCommand{\Nosortcref}[1]{%
  \begingroup\@cref@sortfalse\Cref{#1}\endgroup
}
\DeclareRobustCommand{\nocrefsort}{\@cref@sortfalse}
\makeatother

\newdimen\subfigwidth

\begin{document}

\begin{frontmatter}
    \makeatletter
    \ifdefined\X@thanks\else\xdef\X@thanks{0}\fi
    \ifdefined\X@previous\else\xdef\X@previous{0}\fi
    \makeatother
    \title{Combinatorics of Beacon-based Routing in Three
    Dimensions\tnoteref{thanks}\tnoteref{previous}}
    \tnotetext[thanks]{Supported in part by DFG grant MU 3501/1
    and ERC StG 757609.}
    \tnotetext[previous]{A preliminary version appeared as
     J.~Cleve and W.~Mulzer.
     \emph{Combinatorics of Beacon-based Routing in Three Dimensions}.
      Proc.~13th LATIN, pp.~346--360.}
    \tnotetext[license]{\textcopyright{} 2020. This manuscript version is made available under the CC-BY-NC-ND 4.0 license \texttt{http://creativecommons.org/licenses/by-nc-nd/4.0/}.}
    \author{Jonas Cleve}
    \ead{jonascleve@inf.fu-berlin.de}
    \author{Wolfgang Mulzer}
    \ead{mulzer@inf.fu-berlin.de}
    \address{Institut f\"ur Informatik, Freie Universit\"at
    Berlin, Berlin, Germany}

\begin{abstract}
A beacon $b \in \R^d$ is a point\hyp{}shaped object in
$d$-dimensional space that can 
exert a magnetic pull on any other point\hyp{}shaped 
object $p \in \R^d$. This object $p$ then moves greedily 
towards $b$. The motion stops when $p$ 
gets stuck at an obstacle or when $p$ 
reaches $b$. By placing beacons 
inside a $d$-dimensional polyhedron $P$, we can implement a 
scheme to route point\hyp{}shaped objects between 
any two locations in $P$.  We can also place 
beacons to guard $P$, which means that any 
point\hyp{}shaped object in $P$ can reach at 
least one activated beacon.

The notion of beacon\hyp{}based routing and 
guarding was introduced in 2011 by Biro~\etal~[FWCG'11].
The two\hyp{}dimensional 
setting is discussed in great detail in Biro's 2013 PhD 
thesis~[SUNY-SB'13]. 

Here, we consider combinatorial aspects of beacon routing in 
three dimensions. We show that 
$\floor*{\textfrac{m+1}{3}}$ beacons are always 
sufficient and sometimes necessary to route between 
any two points in a given polyhedron $P$, where 
$m$ is the smallest size of a tetrahedral 
decomposition of $P$. This is one of the first results 
to show that beacon routing is also possible in 
higher dimensions.
\end{abstract}


\end{frontmatter}


\section{Introduction}%
\label{sec:introduction}

Visibility in the presence of obstacles 
is a classic notion in combinatorial and 
computational geometry~\cite{Ghosh07}. 
Given a simple polygon $P$ in the plane, 
two points $p$ and $q$ in $P$ can 
\emph{see each other} if and only if the 
line segment between $p$ and $q$ lies in 
$P$ (considered as a closed region). The 
\emph{visibility region} of a point 
$p \in P$ consists of all points 
$q \in P$ such that $p$ and $q$ can see 
each other. These basic definitions and 
their variants have spawned an active 
subarea of computational geometry, with 
whole textbooks devoted to 
it~\cite{Ghosh07,ORourke87}.

In 2011, Biro~\etal~\cite{BiroGaIwKoMi11} introduced the 
concept of \emph{beacon\hyp{}based} 
visibility, where the objects take a more 
active role. A \emph{beacon} $b \in \R^d$ 
is a point\hyp{}shaped object in 
$d$-dimensional space. The beacon $b$ can 
be \emph{enabled} or \emph{disabled}. 
Once $b$ is enabled, it exerts a 
\emph{magnetic pull} on any other 
point\hyp{}shaped object $p$ in $\R^d$.
Then, the object $p$ moves in the 
direction that most rapidly decreases 
the distance between $b$ and $p$. In the 
simplest case, this motion proceeds along 
the line segment $pb$. If $p$ encounters an 
obstacle that blocks the direct path 
along $pb$, then $p$ slides along the 
boundary of the obstacle in the direction 
that most rapidly decreases the distance 
to $b$.  If this is not possible, the 
motion ends, and we say that $p$ gets 
\emph{stuck}. If $p$ does not get stuck, 
then it reaches $b$, and we say that $p$ 
is \emph{attracted} by $b$.
See \cref{fig:intro:beacon-attraction} for examples.
The \emph{attraction region} of $b$ consists
of all points that are attracted 
by $b$. This is an extension of 
classic visibility: the visibility region 
of $b$ is a subset of the attraction 
region of $b$. However, unlike classic
visibility, beacon attraction is not 
symmetric. Thus, it makes also sense to consider
the \emph{inverse} attraction region of
a point $p$, i.e., the set of all beacon
positions $b$ such that $b$ attracts $p$.
Two examples of these regions can be found in
\cref{fig:intro:attraction-regions}.

\begin{figure}[tp]
    \centering
    \tikzmaybeexternalizenext{intro--beacon-attraction}
    \ifbuildtikz%
    \input{tikz/intro--beacon-attraction.tikz}%
    \else%
    \includegraphics{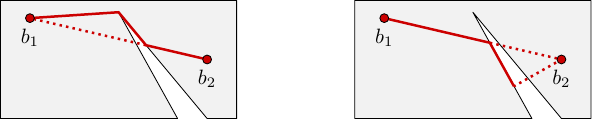}%
    \fi%

    \caption{%
    Attraction is not symmetric. In this two-dimensional example
    $b_1$ attracts $b_2$ (left) but $b_2$ does not attract $b_1$
    (right).}%
    \label{fig:intro:beacon-attraction}
\end{figure}

\begin{figure}[bp]
    \centering
    \tikzmaybeexternalizenext{intro--attraction-regions}
    \ifbuildtikz%
    \input{tikz/intro--attraction-regions.tikz}%
    \else%
    \includegraphics{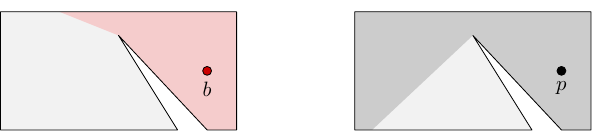}%
    \fi%

    \caption{%
    The \emph{attraction region} of a beacon $b$ (left) and the \emph{inverse attraction region} of a point $p$ (right).}%
    \label{fig:intro:attraction-regions}
\end{figure}

The PhD thesis of Biro~\cite{Biro13} 
constitutes the first in-depth study of 
beacon-based visibility. In particular, 
it considers beacon-based routing and 
guarding in (two\hyp{}dimensional) 
polygonal domains. The idea of 
beacon-based  routing is as follows:
suppose we have a polygonal domain $P$ 
that contains a set $B$ of beacons, and 
suppose we want to route a 
point\hyp{}shaped object $p$ towards a 
target $t$. We assume that 
$t$ can also act as a beacon, even if it 
is not contained in $B$. The 
routing proceeds by successive 
activation of beacons in $B \cup \{t\}$:
a first beacon $b_1 \in B$ is enabled 
to attract $p$ until it reaches $b_1$.
Subsequently, $b_1$ is disabled, and a 
second beacon $b_2 \in B$ is  switched 
on, again attracting $p$ until it 
reaches $b_2$. This is 
repeated until the last (implicit) 
beacon at $t$ is enabled and finally 
attracts $p$ to its location. The 
challenge is to devise a strategy for 
placing the beacons in $P$ and for 
choosing a sequence of beacon 
activations such that it becomes 
possible to route between any two 
locations $s$ and $t$ in $P$. The 
size of $B$ should be minimized. Note 
that we require that every activated 
beacon must attract $p$ until it 
reaches the beacon's location. Only 
then are we allowed to enable the next 
beacon.  Thus, if $p$ gets stuck, the 
process ends and the routing is 
considered to be unsuccessful.

In beacon-based guarding (or coverage), 
the goal is to choose a minimum-size set $B$ 
of beacons such that the union of the 
attraction regions for
$B$ covers the whole polygonal domain 
$P$. In this case, we say that $B$ 
\emph{covers} $P$. This is 
analogous to the classic 
art-gallery problem~\cite{ORourke87}, 
using beacon-based visibility instead 
of straight-line visibility.

\subsection{Related Work}%
\label{sec:related-work}

\paragraph{Two dimensions}
As mentioned above, a large part of 
the pioneering work on beacon-based 
routing and guarding was done by 
Biro and his co-authors~\cite{BiroGaIwKoMi11,BiroIwKoMi13,BiroGaIwKoMi13}.
An extensive collection of results 
can be found in Biro's  
PhD thesis~\cite{Biro13}.

Biro and his co-authors showed that
$\floor*{\textfrac{n}{2}} - 1$ beacons 
always suffice and sometimes are necessary 
for routing in a simple polygon with $n$ 
vertices~\cite[Theorem~1]{BiroGaIwKoMi13}.
We will discuss this result in more 
detail in Section~\ref{sec:preliminaries}.  
More generally, to route in a polygon with 
$n$ vertices and $h$ holes, $\floor*{\textfrac{n}{2}} - h - 1$ 
beacons are sometimes necessary and 
$\floor*{\textfrac{n}{2}} + h - 1$ beacons are 
always sufficient~\cite[Theorem~2]{BiroGaIwKoMi13}.
For \emph{orthogonal} polygons\footnote{A
planar polygon is \emph{orthogonal} if all its
edges are parallel to the $x$- or the $y$-axis.}, 
they showed only a loose lower bound of 
$\floor*{\textfrac{n}{4}} - 1$ beacons, 
leaving a larger gap for the routing 
problem~\cite[Theorem~3]{BiroGaIwKoMi13}.

For beacon\hyp{}based guarding of a simple 
polygon and of a polygon with $h$ holes, 
they showed that $\floor*{\textfrac{4n}{13}}$ 
beacons are sometimes necessary, while 
$\floor*{\textfrac{n + h}{3}}$ beacons are 
always sufficient~\cite[Theorem~5]{BiroGaIwKoMi13}.
In particular, the upper bound for simple polygons 
is $\floor*{\textfrac{n}{3}}$. For orthogonal polygons, 
they obtained an upper bound of 
$\floor*{\textfrac{n}{4}}$ and a lower bound of 
$\floor*{\textfrac{n+4}{8}}$~\cite[Section~6]{BiroGaIwKoMi13}.

Bae~\etal~\cite{BaeShiVi16} improved 
some of these bounds
by showing that $\floor*{\textfrac{n}{6}}$ 
beacons are sometimes needed and always sufficient for 
beacon\hyp{}based guarding in orthogonal polygons.
They also proved that if the polygon is monotone
and orthogonal, the bound reduces to 
$\floor*{\textfrac{n+4}{8}}$.
The gap for routing in simple orthogonal polygons 
was finally closed by Shermer~\cite{Shermer15} who showed
that $\floor*{\textfrac{n - 4}{3}}$ beacons are always sufficient 
and sometimes necessary.

Aldana-Galv\'an~\etal~\cite{AldanaGalvanAlCaNeSoUrVe17b}
extended the notion of coverage to both the
interior and the exterior of a given polygon. They
proved that  
$\floor*{\textfrac{n}{4}}+1$ vertex beacons always suffice to 
simultaneously cover the interior and exterior of an orthogonal 
polygon with $n$ 
vertices (possibly with holes)~\cite[Theorem~1]{AldanaGalvanAlCaNeSoUrVe17b}.
\cref{tab:best-known-2d-beacon-results} gives an
overview of the currently best results for 
routing and guarding in two dimensions.

\begin{table}[t]
    \centering
    \small
    \begin{tabular}{Cl Cl Cl Cl Cl}
        \toprule
        &  & \multicolumn{2}{c}{Bound} &
        \\
        Problem & Polygon type
	& \multicolumn{1}{c}{Lower} & \multicolumn{1}{c}{Upper} & Reference
        \\\midrule
        & Simple &
        \multicolumn{2}{c}{$\floor*{\textfrac{n}{2}} - 1$} &
        \cite[Thm~1]{BiroGaIwKoMi13}
        \\
        Routing & With holes &
        $\floor*{\textfrac{n}{2}} - h - 1$ (*) &
        $\floor*{\textfrac{n}{2}} + h - 1$ & \cite[Thm~2]{BiroGaIwKoMi13}
        \\
        & Orthogonal &
        \multicolumn{2}{c}{$\floor*{\textfrac{n-4}{3}}$} & \cite{Shermer15}
        \\
        \midrule
        & Simple & $\floor*{\textfrac{4n}{13}}$ &
        $\floor*{\textfrac{n}{3}}$ (*) & \cite[Thm~5]{BiroGaIwKoMi13}
        \\
        Guarding & With holes & $\floor*{\textfrac{4n}{13}}$
        & $\floor!{\textfrac{n+h}{3}}$ (*) & \cite[Thm~5]{BiroGaIwKoMi13}
        \\
        & Orthogonal &
        \multicolumn{2}{c}{\(\floor*{\textfrac{n}{6}}\)} &
        \cite{BaeShiVi16}
        \\\bottomrule
    \end{tabular}
    \caption{The currently best results in two dimensions.
    The bounds marked (*) were conjectured to be tight
    by Biro~\cite[Conjectures~6.3.3, 7.3.7, and 7.3.9]{Biro13}}%
    \label{tab:best-known-2d-beacon-results}
\end{table}

So far, we have only discussed results that give
combinatorial bounds on the number of beacons
needed to guard or to route in certain
classes of polygons. Naturally, the notions
of beacon-based routing and guarding also lead
to interesting algorithmic questions.
As is to be expected, several optimization
problems associated with beacons are hard:
Biro~\cite[Theorems~6.2.2, 6.2.3, and 6.2.4]{Biro13}
showed that the \textsc{All-Pair}, \textsc{All-Sink},
and \textsc{All-Source} variants of the optimal
beacon routing problem are NP-hard. In these
problems, we are given a simple polygon $P$, and we
need to find a minimum set $B$ of beacons such
that we can route between any pair of points
in $P$; from a given location $s \in P$
to all other points in $P$; or from all points
in $P$ to a given location $t \in P$, respectively.
Biro also showed that given a simple polygon $P$,
it is NP-hard to find a minimum set of beacons 
that covers $P$~\cite[Theorem~7.2.1]{Biro13}.

On the positive side,
Biro~\etal~\cite[Theorem~6]{BiroIwKoMi13} 
presented an algorithm to compute the attraction
region of a given beacon in a polygon $P$ with $n$
vertices and $h$ holes in $O(n + h  \log^{1 + \eps} h)$
time and $O(n)$ space, for any fixed $\eps > 0$.
They also described how to find the \emph{inverse} attraction
region of a point in a polygon $P$ with $n$ vertices in 
$O(n^2)$ time~\cite[Theorem~8]{BiroIwKoMi13}. More generally,
the inverse attraction region of a polygonal region $R$ in $P$ with
$m$ vertices can be computed in $O(n^2m^2)$ 
time~\cite[Theorem~8]{BiroIwKoMi13}. As for routing,
Biro~\etal show how to find a \emph{minimum-hop-beacon path}
between two points $s$ and $t$ in a polygon with $n$ vertices
and $h$ holes from
a given set of $m$ beacons  in $O(m(n + h\log^{1+\eps} h + m \log h))$
time~\cite[Theorem~11]{BiroIwKoMi13}. They also provide a 
$O(n^3)$-time $2$-approximation algorithm for the case that
the beacons can be placed arbitrarily inside the polygon.
As the authors point out, this approximation algorithm can 
also be applied repeatedly to obtain a PTAS.
More recently, 
Kostitsyna~\etal~\cite{KostitsynaKoLaRa18} gave an optimal
algorithm to compute the inverse beacon
attraction region of a point in a simple
polygon in $O(n \log n)$ time.
Further algorithmic results can be found
in Kouhestani's PhD thesis~\cite{Kouhestani16}.

\paragraph{Three dimensions}

This work is based on the Master's thesis of the first 
author~\cite{Cleve17} who presented the first
combinatorial bounds for 
beacon\hyp{}based routing in three dimensions.
In his thesis, Cleve also 
showed that Biro's NP-hardness and APX-hardness results for 
optimum beacon routing extend to three dimensions, 
by a simple lifting argument~\cite[Section~4.3]{Cleve17}.
Finally, he constructed
a three-dimensional polyhedron that cannot be 
guarded by placing a beacon at every vertex~\cite[Lemma~6.1]{Cleve17}.
Independently, and almost at the same time, 
Aldana-Galv\'an~\etal~\cite[Section~2]{AldanaGalvanAlCaNeSoUrVe17a} 
obtained a stronger result: there exists an 
\emph{orthogonal} polyhedron that
cannot be covered by beacons at every vertex.
Furthermore, Aldana-Galv\'an~\etal~\cite[Theorem~1]{AldanaGalvanAlCaNeSoUrVe17a} 
showed that every \emph{orthotree}\footnote{ 
An orthotree is an orthogonal polyhedron made out of boxes 
that are glued face to face and whose dual graph is a tree.} 
with $n$ vertices can 
be covered by $\floor*{\textfrac{n}{8}}$ beacons. They 
described a family of orthotrees where this number of beacons is 
needed.
They also proved a tight bound of 
$\floor*{\textfrac{n}{12}}$ becons for \emph{well-separated} 
orthotrees.\footnote{An orthotree is well-separated if 
its dual graph has the property that all neighbors of a vertex 
with degree strictly greater than $2$ have degree at most $2$.}
Shortly afterwards, Aldana-Galv\'an~\etal~\cite{AldanaGalvanAlCaNeSoUrVe17b} 
introduced the notion of \emph{edge beacons}. Here,
every point of an edge $e$ may exert a magnet pull on a point-shaped
object $p$, and 
$p$ always moves towards the point on $e$ closest to
it.
Aldana-Galv\'an~\etal~prove 
that $\floor*{\textfrac{m}{12}}$ edge beacons are 
always sufficient and sometimes $\floor*{\textfrac{m}{21}}$ 
edge beacons are necessary to cover an orthogonal polyhedron with $m$ 
edges~\cite[Theorems~3 and 4]{AldanaGalvanAlCaNeSoUrVe17b}.
If both the interior and the exterior of an orthogonal polyhedron 
should be covered simultaneously, $\floor*{\textfrac{m}{6}}$ 
is a tight bound for the number of edge beacons 
required~\cite[Theorem~5]{AldanaGalvanAlCaNeSoUrVe17b}.


\section{Preliminaries}%
\label{sec:preliminaries}

We begin by reviewing the proof that
$\floor*{\textfrac{n}{2}} - 1$ beacons
are needed 
for routing in a simple polygon 
with $n$ vertices~\cite[Theorem~1]{BiroGaIwKoMi13}. 
This serves two purposes: on the one hand,
the argument serves as a starting point for our
three-dimensional bound; on the other hand,
it provides an opportunity to correct a 
slight gap in the published proof by 
Biro~\etal~\cite{BiroGaIwKoMi13}.\footnote{This
issue and a possible fix have also been discovered by
Tom Shermer, a fact personally communicated to us by Irina
Kostitsyna~\cite{Kostitsyna19},
but as
far as we know, no updated version of the proof
has been published to date.}

\subsection{Two-dimensional Upper Bound }

The following theorem states the main result
for beacon-based routing in two dimensions.

\begin{theorem}[{Biro~\etal~\cite[Theorem~1]{BiroGaIwKoMi13}}]
\label{thm:2d:biro-sharp-bound}
Let $P$ be a simple polygon with $n$ 
vertices.  Then, 
$\floor*{\textfrac{n}{2}} - 1$ beacons 
are sometimes necessary and always sufficient 
to route between any two points in $P$.
\end{theorem}

The strategy of Biro~\etal~\cite{BiroGaIwKoMi13} 
is as follows: they triangulate 
$P$ to obtain a partition into 
$n - 2$ triangles. Then, they place 
the beacons in $P$ with an inductive strategy. 
In each step, one beacon $b$
is placed, and at least two 
triangles are removed.
They claim that there is always
a way to position $b$ on the 
boundary of the remaining polygon such 
that the whole interior of the removed 
triangles can be seen from $b$.
The inductive procedure ends as soon as 
no more triangles are left. 
Biro~\etal~conclude 
that $\floor*{\textfrac{n}{2}} -1$ beacons 
suffice for routing. 

The technical heart of the argument lies
in an analysis of different 
triangle configurations. The goal is to show that by 
placing a single beacon, at least two 
triangles can be removed.
One configuration is as follows:\footnote{We follow the notation 
of the original work~\cite{BiroGaIwKoMi13}.}
we have a central triangle 
$\sig{2} = \triangle BCD$ with two 
adjacent triangles $\sig{1} = \triangle ABC$ 
and $\sig{3} = \triangle CDF$. 
Biro~\etal~\cite{BiroGaIwKoMi13} would like
to argue that one can position a beacon 
$b$ on the free edge $BD$ of \sig{2} 
such that the whole polygon $ABDFC$ 
is completely visible to $b$;
see \cref{fig:2d:three-triangles-visibility-from-paper}.
More precisely, their reasoning 
goes like this:

\begin{displaycquote}[p.~2]{BiroGaIwKoMi13}
The location $b$ along $BD$ is chosen so 
the pentagon $ABDFC$ is visible to $b$.
This is always possible, by placing $b$ 
on the correct side of lines $CF$ and 
$AC$.  Then, any point in triangles 
$\triangle ABC$, $\triangle BCD$, 
$\triangle CDF$ can be routed to or 
from $b$ as $b$ is \emph{visible} to 
each point in those triangles.
\end{displaycquote}

However, the condition that $b$ lies
to the right of $AC$ and to the left 
of $FC$ is not sufficient for the whole
pentagon $ABDFC$ to be visible from $b$. 
For this, $b$ must also be to the 
left of $AB$ and to the right of $FD$, 
i.e., in the visibility cone of both 
\sig{1} and \sig{3}.
\Cref{fig:2d:three-triangles-visibility-error} 
shows a situation where this 
cannot be done:
the line through $B$ and $D$ limits 
the visibility of any beacon $b$ 
in the relative interior of the  line 
segment $BD$. Moreover, if we place 
$b$ at $B$ or at $D$, then $b$ still 
cannot see the full pentagon.

\begin{figure}[tp]
    \centering
    \tikzmaybeexternalizenext{2d--three-triangles-visibility-from-paper}
    \ifbuildtikz%
    \input{tikz/2d--three-triangles-visibility-from-paper.tikz}%
    \else%
    \includegraphics{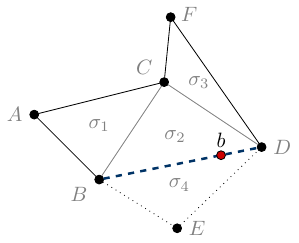}%
    \fi%

    \caption{The situation analyzed by Biro~\etal\cite{BiroGaIwKoMi13}.
             Here, $b$ can be placed near $D$ so that
         $b$ can see every point inside $ABDFC$.
         The edges $AB$, $AC$, $CF$, and $DF$ are boundary edges
         and $BD$ is a diagonal.}%
    \label{fig:2d:three-triangles-visibility-from-paper}
\end{figure}

\begin{figure}[tp]
    \centering
    \subfigwidth=0.47\textwidth
    \captionsetup[subfigure]{width=\subfigwidth,justification=raggedright}
    \tikzmaybeexternalizenext{2d--three-triangles-visibility-error}
    \subfloat[Here, $b$ cannot be placed on $BD$ to 
    see the full pentagon.]{
        \makebox[\subfigwidth][c]{%
    \ifbuildtikz%
    \input{tikz/2d--three-triangles-visibility-error.tikz}%
    \else%
    \includegraphics{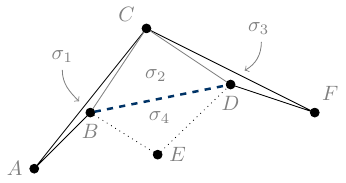}%
    \fi%
}%
        \label{fig:2d:three-triangles-visibility-error}
    }
    \hfill
    \subfigwidth=0.47\textwidth
    \captionsetup[subfigure]{width=\subfigwidth,justification=raggedright}
    \tikzmaybeexternalizenext{2d--three-triangles-attraction}
    \subfloat[No matter where $b$ lies on $BD$, 
    it cannot be attracted by both $A$ and $F$.]{
        \makebox[\subfigwidth][c]{%
    \ifbuildtikz%
    \input{tikz/2d--three-triangles-attraction.tikz}%
    \else%
    \includegraphics{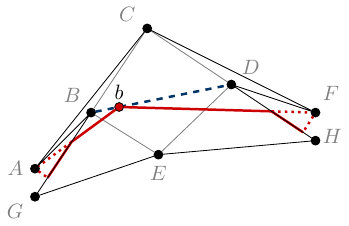}%
    \fi%
}%
        \label{fig:2d:three-triangles-attraction}
    }
    \caption{It is not always possible to place one beacon $b$ on the line 
    segment $BD$ such that it attracts and is attracted by all points 
    inside the pentagon $ABDFC$.}%
    \label{fig:2d:three-triangles-failures}
\end{figure}

Nonetheless, 
visibility is not actually required;
mutual attraction would be
enough for the argument to go through.
In fact, we can always place $b$ so that it 
attracts all points inside the pentagon $ABDFC$. 
Unfortunately, the inverse does not hold.
Consider \cref{fig:2d:three-triangles-attraction}:
unless $b$ is placed at $B$, a point-shaped
object at $b$ that is attracted by $A$ will
get stuck on the line segment $BG$; and analogously for
$D$ and $F$. 
Since $b$ cannot be placed
simultaneously
at both $B$ and 
$D$,
the requirement that $b$ is attracted 
by both $A$ and $F$  cannot be fulfilled.

Nevertheless, Theorem~\ref{thm:2d:biro-sharp-bound} 
still holds, as we will show in the 
following lemma.
For completeness, we present the proof
in full detail, and we indicate where 
we depart from the original argument of
Biro~\etal~\cite[Theorem~1]{BiroGaIwKoMi13}.

\begin{lemma}[Two-dimensional upper bound]%
\label{lem:2d:upper-bound}
Let $P$ be a simple polygon with $n \geq 2$ vertices. 
Then,
$\floor*{\textfrac{n}{2}}-1$ beacons are always sufficient 
to route between any two points in $P$.
\end{lemma}

\begin{proof}

The proof proceeds by induction on $n$.
For the base case, we assume that 
$2 \leq n \leq 4$. If $n \in \{2, 3\}$, 
then $P$ is either a line segment or a single 
triangle. In both cases, $P$ is convex and no 
beacon is needed.
For $n = 4$, we let $d$ be a diagonal of $P$.\footnote{A 
diagonal is a line segment whose endpoints are 
vertices of $P$ and whose relative interior 
lies in the interior of $P$.}
We place one beacon at an arbitrary point $b$ 
on $d$.
Then, every point $p \in P$ can see $b$,
which means that $p$ and $b$ mutually attract.
Thus, we can route from every 
$s \in P$ to every $t \in P$ via $b$.

Now suppose that $n > 4$ and assume that 
Lemma~\ref{lem:2d:upper-bound} holds for all
simple polygons with at most $n - 1$ vertices.
We triangulate $P$ and consider the dual graph 
$T$ of the triangulation: the triangles 
constitute the nodes, and two nodes are 
adjacent if and only if the corresponding 
triangles share an edge in the triangulation.
As $P$ is simple, $T$ is a tree with 
$n - 2$ nodes and maximum degree $3$. We take an arbitrary leaf of
$T$, and we declare it the root. 
Let \sig{1} be a triangle that corresponds
to a deepest leaf in $T$. Let \sig{2} be the
parent triangle of \sig{1}.
There are two cases:

\textbf{Case 1}:
the triangle \sig{1} is the only child of
\sig{2}. Let \sig{3} be the parent triangle of
\sig{2}. Then, the triangles \sig{1}, \sig{2}, 
and \sig{3} share a common vertex $v$, and
we place a beacon $b$ at $v$; 
see \cref{fig:2d:three-triangles-one-beacon}.
Next, we remove from $P$ 
the parts of \sig{1} and \sig{2} 
that do not belong to another triangle of $P$.
This gives a simple polygon $P_1$ 
with $n_1 = n - 2$ vertices.
By the inductive hypothesis, there is a
set $B_1$ of at most 
\[ 
\floor*{\frac{n_1}{2}} - 1 = 
\floor*{\frac{n - 2}{2}} - 1 = 
\floor*{\frac{n}{2}} - 2
\]
beacons that
allows us to route between any two
points in $P_1$.
We set $B = B_1 \cup \{b\}$. 
Then, we have $|B| \leq \floor*{\textfrac{n}{2}} - 1$.

It remains to show that we can use $B$ to route 
between any two points in $P$.
By the inductive hypothesis and because $b$ lies
in \sig{3} which remains in $P_1$, 
we can route between $b$ and any point in $P_1$.
Furthermore, due to convexity of triangles, every point 
$p \in \sig{1} \cup \sig{2}$ can see $b$, and thus 
$p$ can attract $b$ and can be attracted by it.
Hence, we can route between any pair of points 
in $P$ using $B$.

\begin{figure}[tbp]
    \centering
    \subfigwidth=0.235\textwidth
    \captionsetup[subfigure]{width=\subfigwidth,justification=raggedright}
    \tikzmaybeexternalizenext{2d--three-triangles-one-beacon}
    \subfloat[The beacon $b$ covers at least three 
    triangles: \sig{1}, \sig{2}, \sig{3}.]{
        \makebox[\subfigwidth][c]{%
    \ifbuildtikz%
    \input{tikz/2d--three-triangles-one-beacon.tikz}%
    \else%
    \includegraphics{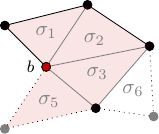}%
    \fi%
}%
        \label{fig:2d:three-triangles-one-beacon}
    }
    \hfill
    \subfigwidth=0.325\textwidth
    \captionsetup[subfigure]{width=\subfigwidth,justification=raggedright}
    \tikzmaybeexternalizenext{2d--four-triangles-two-beacons}
    \subfloat[The two beacons $b_1$ and $b_2$ 
    cover \sig{1}, \sig{2}, \sig{3}, and \sig{4} and 
    both neighbors of \sig{4}.]{
        \makebox[\subfigwidth][c]{%
    \ifbuildtikz%
    \input{tikz/2d--four-triangles-two-beacons.tikz}%
    \else%
    \includegraphics{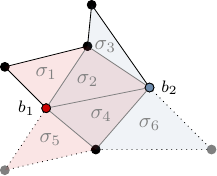}%
    \fi%
}%
        \label{fig:2d:four-triangles-two-beacons}
    }
    \hfill
    \subfigwidth=0.325\textwidth
    \captionsetup[subfigure]{width=\subfigwidth,justification=raggedright}
    \tikzmaybeexternalizenext{2d--four-triangles-peeled-off}
    \subfloat[After removing \sig{1} to \sig{4} two (possibly empty) 
    polygons $P_1$ and $P_2$ remain.]{
        \makebox[\subfigwidth][c]{%
    \ifbuildtikz%
    \input{tikz/2d--four-triangles-peeled-off.tikz}%
    \else%
    \includegraphics{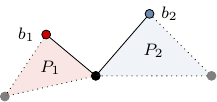}%
    \fi%
}%
        \label{fig:2d:four-triangles-peeled-off}
    }
    \caption{
        The two possible configurations in the inductive step are 
	shown in~\ref{sub@fig:2d:three-triangles-one-beacon} 
	and~\ref{sub@fig:2d:four-triangles-two-beacons}.
        \ref{sub@fig:2d:four-triangles-peeled-off} shows
	the situation 
	of~\ref{sub@fig:2d:four-triangles-two-beacons} after 
	removing the triangles.
    }%
    \label{fig:2d:remove-four-triangles-by-placing-two-beacons}
\end{figure}

\textbf{Case 2}:
the triangle \sig{2} has a second child \sig{3}. 
This is the erroneous case in 
Biro~\etal~\cite[Theorem~1]{BiroGaIwKoMi13}.
Let \sig{4} be the parent triangle of \sig{2}.
Since \sig{1} is a deepest leaf in $T$, if
follows that \sig{3} is also a leaf;
see \cref{fig:2d:four-triangles-two-beacons}.
Instead of placing a single beacon and removing 
three triangles, as suggested 
by Biro~\etal~\cite[Theorem~1]{BiroGaIwKoMi13},
we place two beacons $b_1$, $b_2$ and remove four triangles.
The beacon $b_1$ is placed at the common vertex of 
\sig{1}, \sig{2}, and \sig{4} (marked red), and 
$b_2$ is placed at the common vertex of 
\sig{3}, \sig{2}, and \sig{4} (marked blue).
If \sig{4} has more neighbors, 
they are also covered by $\{b_1, b_2\}$,
see \cref{fig:2d:four-triangles-two-beacons}.

We remove from $P$ the set 
$(\sig{1} \cup \sig{2} \cup \sig{3}) \setminus \{b_1, b_2\}$
and the interior of \sig{4}. This gives two polygons 
$P_1$ and $P_2$ with one common vertex, see
\cref{fig:2d:four-triangles-peeled-off}.
Possibly, $P_1$ or $P_2$ (or both) 
degenerates to a line segment 
from $b_1$ or $b_2$ to the common vertex.
Let $n_1 \geq 2$ be the number of vertices of $P_1$,
and $n_2 \geq 2$ the number of vertices of $P_2$.
We have $n_1 + n_2 = n - 2$, since we 
removed three vertices, and since $P_1$ and $P_2$ 
share one vertex to be counted twice.
As $n_1 \leq n - 1$ and $n_2 \leq n - 1$, 
we can apply the inductive hypothesis to 
$P_1$ and $P_2$. This gives two sets
$B_1 \subset P_1$ and $B_2 \subset P_2$ of 
beacons with $|B_1| \leq 
\floor*{\textfrac{n_1}{2}} - 1$ 
and $|B_2| \leq \floor*{\textfrac{n_2}{2}} - 1$.
We set $B = B_1 \cup B_2 \cup \{b_1, b_2\}$.
Then,
\begin{align*}
|B| &= |B_1| + |B_2| + 2 \leq 
\floor*{\frac{n_1}{2}} - 1 + \floor*{\frac{n_2}{2}} - 1 + 2\\
&=
\floor*{\frac{n_1}{2}} + \floor*{\frac{n_2}{2}}
\leq \floor*{\frac{n_1 + n_2}{2}} =
\floor*{\frac{n - 2}{2}} =
\floor*{\frac{n}{2}} - 1.
\end{align*}

It remains to show that we can route between 
any two points in $P$.
By the inductive hypothesis,
and since $b_1$ lies on the boundary of
$P_1$ and $b_2$ on the boundary of $P_2$, 
we can route between $b_1$ and any point in $P_1$,
and between $b_2$ and any point in $P_2$.
Moreover, since $b_1$ and $b_2$ both lie in \sig{2}, 
they can see and thus attract each other. 
Also, since every removed triangle 
\sig{1}, \sig{2}, \sig{3}, and \sig{4} contains
either $b_1$ or $b_2$, every point in 
$\bigcup_{i = 1}^4 \sig{i}$
can attract and be attracted by $b_1$ or $b_2$. 
It follows that for every point $p \in P$,
we can route between $p$ and $b_1$ or between
$p$ and $b_2$. Since we also can route
between $b_1$ and $b_2$, it follows that
we can route between any two points in
$P$.
\qedhere
\end{proof}

\textbf{Remark.} 
The \emph{extended abstract} for the
original paper by Biro~\etal~from 2011~\cite{BiroGaIwKoMi11},
available on Irina Kostitsyna's ResearchGate profile,
contains an alternative
proof for Theorem~\ref{thm:2d:biro-sharp-bound}.
This version
handles Case~2 slightly differently.
However, we believe that
it is susceptible to
the same issues as the more recent version 
of the proof~\cite{BiroGaIwKoMi13}.
More precisely,
in the alternative proof, the authors use the same notation 
as in
\cref{fig:2d:three-triangles-visibility-from-paper}.
They say that if
$\angle FCB > 3\pi/2$,
the beacon $b$ should be placed at $C$. 
From this, it follows that $\angle CBE \leq 3\pi/2$.
The authors claim that then,
\enquote{all points inside $\triangle BDE$ can reach $b$ and vice
versa}.
However, 
\cref{fig:2d:alternative-proof-error-1} 
shows a case where $E$ cannot 
attract $b$.
A similar counterexample applies for the symmetric case
where $\angle EBC > 3\pi/2$ and
$b$ is placed at $B$.
If both $\angle FCB\leq3\pi/2$ and $\angle EBC\leq 3\pi/2$,
then $b$ is to
be placed \enquote{arbitrarily at either $B$ or $C$},
but \cref{fig:2d:alternative-proof-error-2} shows
a configuration where
both positions cannot be attracted by all points inside the
four triangles.

\begin{figure}[tp]
    \centering
    \subfigwidth=0.47\textwidth
    \captionsetup[subfigure]{width=\subfigwidth,justification=raggedright}
    \tikzmaybeexternalizenext{2d--alternative-proof-error-1}
    \subfloat[%
        $\angle FCB>3\pi/2$ and $b$ is placed at $C$.
        However, despite $\angle EBC\leq 3\pi/2$, a beacon at $E$
        cannot attract an object at $b$.%
    ]{
        \makebox[\subfigwidth][c]{%
    \ifbuildtikz%
    \input{tikz/2d--alternative-proof-error-1.tikz}%
    \else%
    \includegraphics{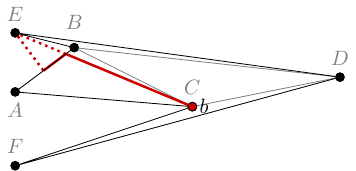}%
    \fi%
}%
        \label{fig:2d:alternative-proof-error-1}
    }
    \hfill
    \subfigwidth=0.47\textwidth
    \captionsetup[subfigure]{width=\subfigwidth,justification=raggedright}
    \tikzmaybeexternalizenext{2d--alternative-proof-error-2}
    \subfloat[%
        $\angle FCB\leq3\pi/2$ and $\angle EBC\leq 3\pi/2$.
        The beacon is to be placed arbitrarily at $B$ or $C$.
        However, for both positions it cannot be attracted by either
        $F$ or $E$.%
    ]{
        \makebox[\subfigwidth][c]{%
    \ifbuildtikz%
    \input{tikz/2d--alternative-proof-error-2.tikz}%
    \else%
    \includegraphics{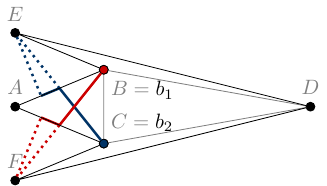}%
    \fi%
}%
        \label{fig:2d:alternative-proof-error-2}
    }
    \caption{Two counterexamples for the alternative proof of Biro~\etal~\cite{BiroGaIwKoMi11}.}%
    \label{fig:2d:alternative-proof-error}
\end{figure}

\subsection{Tetrahedral Decompositions}

To generalize the proof strategy from 
Theorem~\ref{thm:2d:biro-sharp-bound}
to $\R^3$, we need 
a three-dimensional analogue
of polygon triangulation:
the decomposition of a bounded polyhedron into tetrahedra.
This creates several difficulties
that are
not present in the two-dimensional
case.
In 1911, Lennes~\cite{Lennes11} showed that
there are
polyhedra that cannot be decomposed into 
tetrahedra without additional
\emph{Steiner points}. 
In fact, it is \NP-complete
to decide whether 
a tetrahedral decomposition without 
Steiner points exists~\cite{RuppertSe92}.
The \emph{size} of a tetrahedral decomposition
is the number of tetrahedra contained in it.
Unlike in two dimensions,
the size of a tetrahedral decomposition
may significantly exceed
the number of vertices in the polyhedron.
Chazelle~\cite{Chazelle84} showed that 
for any $n$, there exists a 
polyhedron with $\Theta(n)$ vertices for 
which any decomposition into convex
parts needs at least $\Omega(n^2)$ pieces.
On the other hand, Bern and 
Eppstein~\cite[Theorem~13]{BernEp95} described
how to decompose any polyhedron 
into $O(n^2)$ tetrahedra using 
$O(n^2)$ Steiner points.
Furthermore, a 
tetrahedral decomposition clearly must have
size at least 
$n - 3$.
A single  polyhedron may have 
different tetrahedral decompositions 
of varying sizes.
For example, the triangular bipyramid 
can be decomposed into two or three 
tetrahedra~\cite[p.~228]{RuppertSe92}.
Thus, our bounds 
will be in terms of the minimum size of a 
decomposition rather than the number of vertices.
Steiner points are allowed.

To extend the ideas for two dimensions to
$\R^3$, we must understand the 
dual graph of a tetrahedral decomposition.
This graph is defined as follows:

\begin{definition}
\label{def:3d-tetrahedra:dual-graph}
Given a tetrahedral decomposition
$\Sigma = \{\sig{1},\dots,\sig{m}\}$
of a three-dimensional polyhedron, 
the \emph{dual graph} $D(\Sigma)$ of
$\Sigma$ is the undirected graph 
with vertex set 
$\{\sig{1}, \dots, \sig{m}\}$ 
in which there is an edge between
two distinct tetrahedra \sig{i} and \sig{j}
if and only if 
\sig{i} and \sig{j} share
a triangular facet.
\end{definition}

Similarly to the
two-dimensional case, 
the dual graph  $D(\Sigma)$
of a tetrahedral decomposition
has maximum degree $4$.
However, unlike in
two dimensions, $D(\Sigma)$ is
not necessarily a tree.
The following lemma provides
a tool for placing beacons in 
connected subgraphs of $D(\Sigma)$.

\begin{lemma}\label{lem:3d-tetrahedra:x-tetrahedra-share-y}
Let $\Sigma$ be a tetrahedral 
decomposition of a 
three-dimensional polyhedron, and
let $D(\Sigma)$ be the dual graph
of $\Sigma$. Consider a
set $ S\subseteq \Sigma$ of tetrahedra such that
the induced subgraph $D(S)$ of $D(\Sigma)$ is 
connected.
Then,
\begin{theoremcases}
\item\label{lem:3d-tetrahedra:x-tetrahedra-share-y:two-share-face}
if $\abs{S} = 2$,
the tetrahedra in $S$ share a triangular 
facet;
\item\label{lem:3d-tetrahedra:x-tetrahedra-share-y:three-share-edge}
if $\abs{S} = 3$, the tetrahedra in $S$ 
share one edge; and
\item\label{lem:3d-tetrahedra:x-tetrahedra-share-y:four-share-vertex}
if  $\abs{S} = 4 $, the tetrahedra in 
$S$ share at least one vertex.
\end{theoremcases}
\end{lemma}
\begin{proof}
We consider the three cases separately.

\textbf{Case (i)}:
this follows directly from 
\cref{def:3d-tetrahedra:dual-graph}.

\textbf{Case (ii)}:
since $D(S)$ is connected and since
$|S| = 3$, there is a tetrahedron
$\sigma \in S$ adjacent to
the other two.
By \cref{def:3d-tetrahedra:dual-graph}, 
this means that
$\sigma$ shares a 
facet with each of the other two tetrahedra.
Since $\sigma$ is a tetrahedron,
any two facets in $\sigma$ share 
an edge. The claim follows.

\textbf{Case (iii)}: see
\cref{fig:3d-tetrahedra:configurations-of-four-tetrahedra}.
Let $S' \subset S$ be 
three tetrahedra in $S$ so that $D(S')$ 
is connected.
By~(ii), 
the tetrahedra in $S'$ share
an edge $e$.
By \cref{def:3d-tetrahedra:dual-graph}, the remaining 
tetrahedron in $S \setminus S'$ 
shares a facet $f$ with a
tetrahedron $\sigma \in S'$.
Since $e$ contains two vertices
of $\sigma$ while $f$ contains three vertices, 
$e$ and $f$ must share at least one vertex.
The claim follows.
\end{proof}

\begin{figure}[tbp]
\centering

\subfigwidth=0.29\textwidth
\captionsetup[subfigure]{width=\subfigwidth,justification=raggedright}
\tikzmaybeexternalizenext{3d-tetrahedra--configurations-of-four-tetrahedra--star}
\subfloat[One tetrahedron in the center has all 
other tetrahedra as neighbors.]{
\makebox[\subfigwidth][c]{%
    \ifbuildtikz%
    \input{tikz/3d-tetrahedra--configurations-of-four-tetrahedra--star.tikz}%
    \else%
    \includegraphics{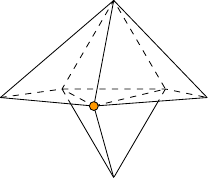}%
    \fi%
}%
        \label{fig:3d-tetrahedra:configurations-of-four-tetrahedra:star}
    }
\hfill
    \subfigwidth=0.29\textwidth
    \captionsetup[subfigure]{width=\subfigwidth,justification=raggedright}
    \tikzmaybeexternalizenext{3d-tetrahedra--configurations-of-four-tetrahedra--line}
    \subfloat[Two tetrahedra with one and two 
    tetrahedra with two neighbors.]{
        \makebox[\subfigwidth][c]{%
    \ifbuildtikz%
    \input{tikz/3d-tetrahedra--configurations-of-four-tetrahedra--line.tikz}%
    \else%
    \includegraphics{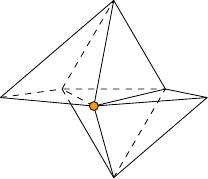}%
    \fi%
}%
        \label{fig:3d-tetrahedra:configurations-of-four-tetrahedra:line}
    }
    \hfill
    \subfigwidth=0.29\textwidth
    \captionsetup[subfigure]{width=\subfigwidth,justification=raggedright}
    \tikzmaybeexternalizenext{3d-tetrahedra--configurations-of-four-tetrahedra--line-shared-edge}
    \subfloat[All four tetrahedra share one edge.]{
        \makebox[\subfigwidth][c]{%
    \ifbuildtikz%
    \input{tikz/3d-tetrahedra--configurations-of-four-tetrahedra--line-shared-edge.tikz}%
    \else%
    \includegraphics{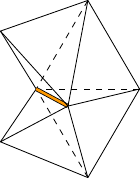}%
    \fi%
}%
        \label{fig:3d-tetrahedra:configurations-of-four-tetrahedra:line-shared-edge}
    }
    \caption{%
       The three possible configuration for a
       polyhedron with a decomposition 
	into four tetrahedra.
        The shared vertex or edge is marked.
    }%
    \label{fig:3d-tetrahedra:configurations-of-four-tetrahedra}
\end{figure}


\section{An Upper Bound for Beacon\hyp{}based Routing}%
\label{sec:3d-tetrahedra:upper-bound}

We now give an upper bound on the number 
of beacons needed to route within a polyhedron,
extending the
strategy of Biro~\etal~\cite{BiroGaIwKoMi13},
as described in Section~\ref{sec:preliminaries},
to three dimensions.
We want to show the following:
\begin{theorem}%
\label{lem:3d-tetrahedra:upper-bound}
Let $P$ be a three-dimensional polyhedron,
and let $\Sigma$ be a tetrahedral 
decomposition of $P$ of size $m$.
There is a set
of at most $\floor*{\textfrac{m+1}{3}}$ 
beacons that allows us to route between any 
pair of points in $P$.
\end{theorem}
The rest of this section 
is dedicated to the inductive proof of 
Theorem~\ref{lem:3d-tetrahedra:upper-bound}.
The following lemma constitutes the base
case of the induction.

\begin{lemma}[Base case]%
\label{lem:3d-tetrahedra:upper-bound:base-case}
Let $P$ be a three-dimensional polyhedron,
and let $\Sigma$ be a tetrahedral 
decomposition of $P$ of size $m \leq 4$.
There is a set
of at most $\floor*{\textfrac{m+1}{3}}$ 
beacons that allows us to route between any 
pair of points in $P$.
\end{lemma}

\begin{proof}
If $m = 1$, then $P$ is a 
convex tetrahedron,  and no beacon is needed.
If $m \in \{2, 3, 4\}$, we apply 
\cref{lem:3d-tetrahedra:x-tetrahedra-share-y} 
to obtain a vertex $v$ that is common to
all tetrahedra in $\Sigma$.
We place one beacon $b$ at $v$.
By convexity, every point in $P$ can 
attract and be attracted by $b$,
and the claim follows.
\end{proof}

We proceed to the inductive step. 
For this, we consider a
tetrahedral decomposition $\Sigma$ 
of size $m > 4$.
Our goal is to place $k$ beacons,
for some $k \geq 1$, such that the
beacons
lie in at least $3k + 1$ 
tetrahedra and therefore can attract
and can be attracted by all points in 
those tetrahedra.
Then, we remove at least $3k$ 
tetrahedra, leaving a polyhedron with 
a tetrahedral decomposition of
size strictly less than $m$. We apply 
induction, 
and then show how to route between 
the smaller polyhedron and the removed tetrahedra.

\begin{figure}[tb]
    \centering
    \subfigwidth=0.29\textwidth
    \captionsetup[subfigure]{width=\subfigwidth,justification=raggedright}
    \tikzmaybeexternalizenext{3d-tetrahedra--upper-bound--1_3_4-2s}
    \subfloat[Remove \sig{1}, \sig{3}, and \sig{4} by placing a beacon where all four tetrahedra meet.]{
        \makebox[\subfigwidth][c]{%
    \ifbuildtikz%
    \input{tikz/3d-tetrahedra--upper-bound--1_3_4-2s.tikz}%
    \else%
    \includegraphics{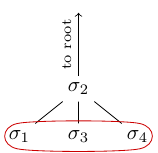}%
    \fi%
}%
        \label{fig:3d-tetrahedra:upper-bound:1_3_4-2*}
    }
    \hfill
    \subfigwidth=0.29\textwidth
    \captionsetup[subfigure]{width=\subfigwidth,justification=raggedright}
    \tikzmaybeexternalizenext{3d-tetrahedra--upper-bound--1_3-2-4s}
    \subfloat[Remove \sig{1}, \sig{2}, and \sig{3} by placing a beacon where all four tetrahedra meet.]{
        \makebox[\subfigwidth][c]{%
    \ifbuildtikz%
    \input{tikz/3d-tetrahedra--upper-bound--1_3-2-4s.tikz}%
    \else%
    \includegraphics{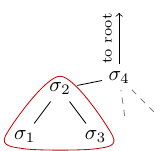}%
    \fi%
}%
        \label{fig:3d-tetrahedra:upper-bound:1_3-2-4*}
    }
    \hfill
    \subfigwidth=0.29\textwidth
    \captionsetup[subfigure]{width=\subfigwidth,justification=raggedright}
    \tikzmaybeexternalizenext{3d-tetrahedra--upper-bound--1-2-3-4s}
    \subfloat[Remove \sig{1}, \sig{2}, and \sig{3} by placing a beacon where all four tetrahedra meet.]{
        \makebox[\subfigwidth][c]{%
    \ifbuildtikz%
    \input{tikz/3d-tetrahedra--upper-bound--1-2-3-4s.tikz}%
    \else%
    \includegraphics{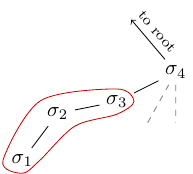}%
    \fi%
}%
        \label{fig:3d-tetrahedra:upper-bound:1-2-3-4*}
    }

    \vspace{1em}

    \subfigwidth=0.29\textwidth
    \captionsetup[subfigure]{width=\subfigwidth,justification=raggedright}
    \tikzmaybeexternalizenext{3d-tetrahedra--upper-bound--1-2_4-3s}
    \subfloat[Remove \sig{1}, \sig{2}, and \sig{4} by placing a beacon where all four tetrahedra meet.]{
        \makebox[\subfigwidth][c]{%
    \ifbuildtikz%
    \input{tikz/3d-tetrahedra--upper-bound--1-2_4-3s.tikz}%
    \else%
    \includegraphics{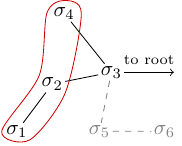}%
    \fi%
}%
        \label{fig:3d-tetrahedra:upper-bound:1-2_4-3*}
    }
    \hfill
    \subfigwidth=0.29\textwidth
    \captionsetup[subfigure]{width=\subfigwidth,justification=raggedright}
    \tikzmaybeexternalizenext{3d-tetrahedra--upper-bound--1_5_7-2_4_6-3s}
    \subfloat[Remove \sig{1}, \sig{2}, \sig{4}, and \sig{5} by placing a beacon where \sig{1} to \sig{5} meet.]{
        \makebox[\subfigwidth][c]{%
    \ifbuildtikz%
    \input{tikz/3d-tetrahedra--upper-bound--1_5_7-2_4_6-3s.tikz}%
    \else%
    \includegraphics{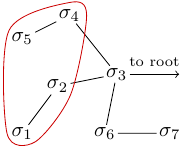}%
    \fi%
}%
        \label{fig:3d-tetrahedra:upper-bound:1_5_7-2_4_6-3*}
    }
    \hfill
    \subfigwidth=0.29\textwidth
    \captionsetup[subfigure]{width=\subfigwidth,justification=raggedright}
    \tikzmaybeexternalizenext{3d-tetrahedra--upper-bound--1_5-2_4-3-6s}
    \subfloat[The number and configuration of \sig{6}'s children must be looked at.]{
        \makebox[\subfigwidth][c]{%
    \ifbuildtikz%
    \input{tikz/3d-tetrahedra--upper-bound--1_5-2_4-3-6s.tikz}%
    \else%
    \includegraphics{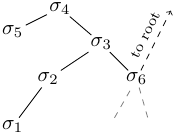}%
    \fi%
}%
        \label{fig:3d-tetrahedra:upper-bound:1_5-2_4-3-6*}
    }
    \caption{The possible configurations in the first part of the inductive step.}%
    \label{fig:3d-tetrahedra:upper-bound}
\end{figure}

To do this, 
we look at the dual graph $D(\Sigma)$ 
of $\Sigma$, as in 
Definition~\ref{def:3d-tetrahedra:dual-graph}.
Let $T$ be a spanning tree of 
$D(\Sigma)$, rooted at an arbitrary 
leaf.
We do not
distinguish between nodes of $T$
and the corresponding
tetrahedra.
Let \sig{1} be a deepest leaf of $T$.
If there are multiple  such leaves, 
we choose \sig{1} such that its parent \sig{2} has 
the largest number of children, breaking 
ties arbitrarily.
\cref{fig:3d-tetrahedra:upper-bound} shows
the
different cases how $T$ can look like
around \sig{1} and \sig{2}.
First, we focus on 
\cref{fig:3d-tetrahedra:upper-bound:1_3_4-2*,fig:3d-tetrahedra:upper-bound:1_3-2-4*,fig:3d-tetrahedra:upper-bound:1-2-3-4*,fig:3d-tetrahedra:upper-bound:1-2_4-3*,fig:3d-tetrahedra:upper-bound:1_5_7-2_4_6-3*}.
In all five cases, $T$ must have at least one additional root 
node---either because $m \geq 5$
or because 
$T$ is  rooted at a leaf.
The situation in \cref{fig:3d-tetrahedra:upper-bound:1_5-2_4-3-6*} 
will be dealt with in \cref{lem:3d-tetrahedra:upper-bound:inductive-step-ii}.

\begin{lemma}[Inductive step I]%
\label{lem:3d-tetrahedra:upper-bound:inductive-step-i}
Let $P$ be a three-dimensional polyhedron, 
and $\Sigma$ a tetrahedral
decomposition of $P$ of size
$m \geq 5$. Let $T$ be a spanning
tree of the dual graph $D(\Sigma)$,
rooted at a leaf of $T$.
Let \sig{1} be a deepest leaf of $T$ 
with the maximum number of siblings, 
and \sig{2} its parent.
Assume that one of the following 
conditions holds:
\begin{theoremcases}
\item\label{lem:3d-tetrahedra:upper-bound:inductive-step-i:case1}
\sig{2} has exactly three children 
\sig{1}, \sig{3}, and \sig{4} 
(see \cref{fig:3d-tetrahedra:upper-bound:1_3_4-2*});
\item\label{lem:3d-tetrahedra:upper-bound:inductive-step-i:case2}
\sig{2} has exactly two children \sig{1} and \sig{3}, 
and a parent \sig{4} (\cref{fig:3d-tetrahedra:upper-bound:1_3-2-4*});
\item\label{lem:3d-tetrahedra:upper-bound:inductive-step-i:case3}
\sig{2} has exactly one child \sig{1} and is 
the only child of its parent \sig{3}, 
whose parent is \sig{4} (\cref{fig:3d-tetrahedra:upper-bound:1-2-3-4*});
\item\label{lem:3d-tetrahedra:upper-bound:inductive-step-i:case4}
\sig{2} has exactly one child \sig{1} and its parent 
\sig{3} has two or three children at least one of which,
say \sig{4}, is a leaf (\cref{fig:3d-tetrahedra:upper-bound:1-2_4-3*}); 
or
\item\label{lem:3d-tetrahedra:upper-bound:inductive-step-i:case5}
\sig{2} has exactly one child \sig{1} and its parent \sig{3} 
has three children, each of which has a single leaf 
child (\cref{fig:3d-tetrahedra:upper-bound:1_5_7-2_4_6-3*}).
\end{theoremcases}
Then, we can place a beacon $b$ at a vertex of 
\sig{1} such that $b$ lies in at least four tetrahedra.
After that, we can remove at least three of these tetrahedra
so that $T$ stays a tree and at
least one remaining tetrahedron in $T$ contains $b$.
\end{lemma}

\begin{proof}
We consider the cases individually.

\textbf{Cases (i--iv)}:
in each case, the induced 
subgraph on $\{\sig{1}, \sig{2}, \sig{3}, \sig{4}\}$ 
is connected.  Thus, 
\cref{lem:3d-tetrahedra:x-tetrahedra-share-y}\ref{lem:3d-tetrahedra:x-tetrahedra-share-y:four-share-vertex} 
implies that the four tetrahedra share a
vertex $v$. We place $b$ at $v$. 
After that, we remove either \sig{1}, \sig{3}, and \sig{4} 
(\cref{lem:3d-tetrahedra:upper-bound:inductive-step-i:case1}); 
\sig{1}, \sig{2}, and \sig{3} 
(\cref{lem:3d-tetrahedra:upper-bound:inductive-step-i:case2,lem:3d-tetrahedra:upper-bound:inductive-step-i:case3}); 
or \sig{1}, \sig{2}, and \sig{4} 
(\cref{lem:3d-tetrahedra:upper-bound:inductive-step-i:case4}).
In each case, we remove either only leaves or 
inner nodes with all their children. 
This means that the tree structure of $T$ is preserved.
Moreover, we only remove three of the four tetrahedra that 
contain $b$, so one of them remains in $T$.

\textbf{Case (v)}:
as shown in \cref{fig:3d-tetrahedra:upper-bound:1_5_7-2_4_6-3*}, 
we have three connected sets, each containing \sig{3}, 
a child \sig{i} of \sig{3}, and \sig{i}'s child:
$\{\sig1, \sig2, \sig3\}$, $\{\sig5, \sig4, \sig3\}$, 
and $\{\sig7,\sig6,\sig3\}$.
By \cref{lem:3d-tetrahedra:x-tetrahedra-share-y}\ref{lem:3d-tetrahedra:x-tetrahedra-share-y:three-share-edge}, 
each set has a common edge. 
These three edges all occur in
\sig{3}, and since \sig{3} is a tetrahedron,
at least two of them share
an endpoint $v$.
Without loss of generality, let these be the common edges of 
$\{\sig1, \sig2, \sig3\}$ and of $\{\sig5, \sig4, \sig3\}$.
We place $b$ at $v$, and we remove 
\sig{1}, \sig{2}, \sig{4}, and \sig{5}.
The beacon $b$ is also contained in \sig{3}, which remains in $T$.
\end{proof}

The final configuration 
is shown in \cref{fig:3d-tetrahedra:upper-bound:1_5-2_4-3-6*}.
The following lemma provides an analysis of how the
tetrahedra can intersect in this case.

\begin{lemma}%
\label{lem:3d-tetrahedra:central-tetrahedron-1-subtree-with-5}
Let $\Sigma$ be a tetrahedral decomposition
of size $6$, and suppose that 
$D(\Sigma)$ has a spanning tree as
in \cref{fig:3d-tetrahedra:central-tetrahedron-1-subtree-with-5:dual-graph}.
Then at least one of the following holds:
\begin{theoremcases}
\item
\label{lem:3d-tetrahedra:central-tetrahedron-1-subtree-with-5:shared-vertex}
\sig{1}, \sig{2}, \sig{3}, \sig{4}, and \sig{5} have a common vertex; or
\item
\label{lem:3d-tetrahedra:central-tetrahedron-1-subtree-with-5:shared-vertex-and-shared-edge}
\sig{3}, \sig{4}, \sig{5}, and \sig{6} share a common vertex 
$v$; \sig{1}, \sig{2}, \sig{3}, and \sig{6} share a common 
edge $e$; and $v \cap e = \emptyset$. A symmetric situation is
also possible.
\end{theoremcases}
\end{lemma}
\begin{figure}[tp]
    \centering
    \subfigwidth=0.2\textwidth
    \captionsetup[subfigure]{width=\subfigwidth,justification=raggedright}
    \tikzexternalizenext{3d-tetrahedra--central-tetrahedron-1-subtree-with-5--dual-graph}
    \subfloat[The dual graph of the tetrahedral decomposition.]{
        \makebox[\subfigwidth][c]{%
    \ifbuildtikz%
    \input{tikz/3d-tetrahedra--central-tetrahedron-1-subtree-with-5--dual-graph.tikz}%
    \else%
    \includegraphics{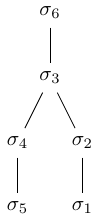}%
    \fi%
}%
        \label{fig:3d-tetrahedra:central-tetrahedron-1-subtree-with-5:dual-graph}
    }
    \hfill
    \subfigwidth=0.35\textwidth
    \captionsetup[subfigure]{width=\subfigwidth,justification=raggedright}
    \tikzmaybeexternalizenext{3d-tetrahedra--central-tetrahedron-1-subtree-with-5--3d-case2}
    \subfloat[The four tetrahedra on the left share a common vertex
    while the four tetrahedra on the right share a common edge.]{
\makebox[\subfigwidth][c]{%
    \ifbuildtikz%
    \input{tikz/3d-tetrahedra--central-tetrahedron-1-subtree-with-5--3d-case2.tikz}%
    \else%
    \includegraphics{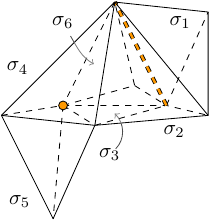}%
    \fi%
}%
        \label{fig:3d-tetrahedra:central-tetrahedron-1-subtree-with-5:3d-case2}
    }
    \hfill
    \subfigwidth=0.32\textwidth
    \captionsetup[subfigure]{width=\subfigwidth,justification=raggedright}
    \tikzmaybeexternalizenext{3d-tetrahedra--central-tetrahedron-1-subtree-with-5--3d-case1}
    \subfloat[All tetrahedra but the rearmost tetrahedron 
    \sig{6} share one common vertex, marked in orange.]{
\makebox[\subfigwidth][c]{%
    \ifbuildtikz%
    \input{tikz/3d-tetrahedra--central-tetrahedron-1-subtree-with-5--3d-case1.tikz}%
    \else%
    \includegraphics{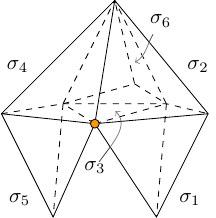}%
    \fi%
}%
        \label{fig:3d-tetrahedra:central-tetrahedron-1-subtree-with-5:3d-case1}
    }
    \caption{
\nocrefsort
A tetrahedron \sig{6} with a subtree of five tetrahedra.
Figures~\ref{sub@fig:3d-tetrahedra:central-tetrahedron-1-subtree-with-5:3d-case2} 
and~\ref{sub@fig:3d-tetrahedra:central-tetrahedron-1-subtree-with-5:3d-case1} 
depict configurations that satisfy 
\cref{lem:3d-tetrahedra:central-tetrahedron-1-subtree-with-5:shared-vertex-and-shared-edge,lem:3d-tetrahedra:central-tetrahedron-1-subtree-with-5:shared-vertex} of \cref{lem:3d-tetrahedra:central-tetrahedron-1-subtree-with-5}, 
respectively.
}%
\label{fig:3d-tetrahedra:central-tetrahedron-1-subtree-with-5}
\end{figure}
\begin{proof}
Let $S_1 = \{\sig{3}, \sig{4}, \sig{5}, \sig{6}\}$ 
and $S_2 = \{\sig{1}, \sig{2}, \sig{3}, \sig{6}\}$.
By \cref{lem:3d-tetrahedra:x-tetrahedra-share-y},
each set shares at least a vertex, but it may also share an edge.
There are three cases: 

\textbf{Case 1}: both $S_1$ and $S_2$ share an edge.
These edges must belong to the triangular facet 
that connects \sig{3} and \sig{6}.
Thus, they share a common vertex,
and~\ref{lem:3d-tetrahedra:central-tetrahedron-1-subtree-with-5:shared-vertex} 
holds.

\textbf{Case 2}:
exactly one of $S_1$, $S_2$ shares an edge $e$,
while the other shares only a vertex $v$. 
If $v \cap e = v$, 
then~\ref{lem:3d-tetrahedra:central-tetrahedron-1-subtree-with-5:shared-vertex} 
applies, and if
$v \cap e = \emptyset$, 
then~\ref{lem:3d-tetrahedra:central-tetrahedron-1-subtree-with-5:shared-vertex-and-shared-edge}
holds---see 
\cref{fig:3d-tetrahedra:central-tetrahedron-1-subtree-with-5:3d-case2}
for an example.

\textbf{Case 3}: both $S_1$ and $S_2$ share only
a vertex; see
\cref{fig:3d-tetrahedra:central-tetrahedron-1-subtree-with-5:3d-case1}.
Let  $v$ be the vertex of \sig{3} 
that is not in the facet shared by \sig{3} and \sig{6}. 
In 
\cref{fig:3d-tetrahedra:central-tetrahedron-1-subtree-with-5:3d-case1},
$v$ is marked orange.
Since \sig{2} is adjacent to \sig{3}, it follows that
\sig{2} contains $v$ and three of its four facets contain 
$v$. 
One of these facets is the shared facet with \sig{3}, 
and we claim that \sig{1} is placed at one of the
other two.
Indeed, \sig{1} cannot be located at the fourth facet of \sig{2}, 
since otherwise it would share an edge with \sig2, \sig{3} and \sig{6},
which is ruled out by the current case. 
Thus, $v \in \sig{1}$, and a symmetric argument 
shows that $v \in \sig{5}$.
It follows 
that~\ref{lem:3d-tetrahedra:central-tetrahedron-1-subtree-with-5:shared-vertex} 
holds.
\end{proof}

Now, we can proceed with
the inductive step for the configuration from
\cref{fig:3d-tetrahedra:upper-bound:1_5-2_4-3-6*}.
The problem is that to remove 
$\{\sig{1}, \dots, \sig{5}\}$, we need two beacons.
However, this does not meet our goal of handling at 
least $3k$ tetrahedra by placing $k$ beacons, for a
$k \geq 1$.
If we removed \sig{6} and if \sig{6} had additional 
children, the remaining dual graph might no longer
be connected, and we could not continue with our
induction. Thus, we must look at 
the (additional) subtrees of \sig{6}.

Since there are many possibilities, 
we wrote a short Python program to
generate all the cases. 
Our program enumerates all rooted, ordered,
ternary trees of height at most three. 
To each such tree, the program repeatedly applies
\cref{lem:3d-tetrahedra:upper-bound:inductive-step-i} 
to prune subtrees. If this
results in an empty tree, the case does not
need to be considered. If not, we save
the remaining tree for manual consideration, eliminating
isomorphic copies of the same tree.
The source code is in~\ref{appendix:trees-program}.
The program leaves us with nine different cases,
shown in 
\cref{fig:3d-tetrahedra:upper-bound-2}.
In each case,
the subtree from 
\cref{fig:3d-tetrahedra:upper-bound:1_5-2_4-3-6*} is present.
The following lemma explains how to place the beacons.

\begin{lemma}[Inductive step II]%
\label{lem:3d-tetrahedra:upper-bound:inductive-step-ii}
Let $P$ be a three-dimensional polyhedron, with a
tetrahedral decomposition $\Sigma$ of size
$m \geq 5$.
Let $T$ be a spanning tree of the dual graph $D(\Sigma)$, rooted 
at an arbitrary leaf.
Let $T' \subseteq T$ be a subtree of $T$ with height $3$ 
for which \cref{lem:3d-tetrahedra:upper-bound:inductive-step-i} 
cannot be applied; see \cref{fig:3d-tetrahedra:upper-bound-2}.

Then, there is a set $B$ of $k \geq 2$ beacons that
are vertices in at least $3k + 1$ 
tetrahedra from $T'$, such that the induced subgraph for $B$ on
$\Sigma$ is connected.
Furthermore, we can remove at least $3k$ 
tetrahedra, each containing a beacon
from $B$, so that $T$ remains connected and 
so that at least one remaining tetrahedron 
contains a beacon from $B$
\end{lemma}
\begin{figure}[tbp]
    \tikzset{every picture/.append style={yscale=0.475,xscale=0.8}}
    \centering
    \subfigwidth=0.13\textwidth
    \captionsetup[subfigure]{width=\subfigwidth,justification=raggedright}
    \tikzmaybeexternalizenext{3d-tetrahedra--upper-bound-2--case1}
    \subfloat[]{
        \makebox[\subfigwidth][c]{%
    \ifbuildtikz%
    \input{tikz/3d-tetrahedra--upper-bound-2--case1.tikz}%
    \else%
    \includegraphics{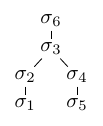}%
    \fi%
}%
        \label{fig:3d-tetrahedra:upper-bound-2:case1}
    }
    \hfill
    \subfigwidth=0.17\textwidth
    \captionsetup[subfigure]{width=\subfigwidth,justification=raggedright}
    \tikzmaybeexternalizenext{3d-tetrahedra--upper-bound-2--case2}
    \subfloat[]{
        \makebox[\subfigwidth][c]{%
    \ifbuildtikz%
    \input{tikz/3d-tetrahedra--upper-bound-2--case2.tikz}%
    \else%
    \includegraphics{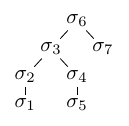}%
    \fi%
}%
        \label{fig:3d-tetrahedra:upper-bound-2:case2}
    }
    \hfill
    \subfigwidth=0.23\textwidth
    \captionsetup[subfigure]{width=\subfigwidth,justification=raggedright}
    \tikzmaybeexternalizenext{3d-tetrahedra--upper-bound-2--case3}
    \subfloat[]{
        \makebox[\subfigwidth][c]{%
    \ifbuildtikz%
    \input{tikz/3d-tetrahedra--upper-bound-2--case3.tikz}%
    \else%
    \includegraphics{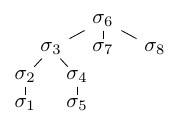}%
    \fi%
}%
        \label{fig:3d-tetrahedra:upper-bound-2:case3}
    }
    \hfill
    \subfigwidth=0.2\textwidth
    \captionsetup[subfigure]{width=\subfigwidth,justification=raggedright}
    \tikzmaybeexternalizenext{3d-tetrahedra--upper-bound-2--case4}
    \subfloat[]{
        \makebox[\subfigwidth][c]{%
    \ifbuildtikz%
    \input{tikz/3d-tetrahedra--upper-bound-2--case4.tikz}%
    \else%
    \includegraphics{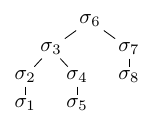}%
    \fi%
}%
        \label{fig:3d-tetrahedra:upper-bound-2:case4}
    }

    \subfigwidth=0.27\textwidth
    \captionsetup[subfigure]{width=\subfigwidth,justification=raggedright}
    \tikzmaybeexternalizenext{3d-tetrahedra--upper-bound-2--case5}
    \subfloat[]{
        \makebox[\subfigwidth][c]{%
    \ifbuildtikz%
    \input{tikz/3d-tetrahedra--upper-bound-2--case5.tikz}%
    \else%
    \includegraphics{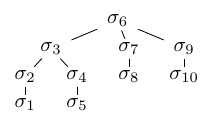}%
    \fi%
}%
        \label{fig:3d-tetrahedra:upper-bound-2:case5}
    }
    \hfill
    \subfigwidth=0.27\textwidth
    \captionsetup[subfigure]{width=\subfigwidth,justification=raggedright}
    \tikzmaybeexternalizenext{3d-tetrahedra--upper-bound-2--case6}
    \subfloat[]{
        \makebox[\subfigwidth][c]{%
    \ifbuildtikz%
    \input{tikz/3d-tetrahedra--upper-bound-2--case6.tikz}%
    \else%
    \includegraphics{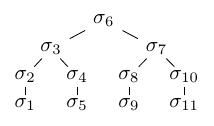}%
    \fi%
}%
        \label{fig:3d-tetrahedra:upper-bound-2:case6}
    }
    \hfill
    \subfigwidth=0.29\textwidth
    \captionsetup[subfigure]{width=\subfigwidth,justification=raggedright}
    \tikzmaybeexternalizenext{3d-tetrahedra--upper-bound-2--case7}
    \subfloat[]{
        \makebox[\subfigwidth][c]{%
    \ifbuildtikz%
    \input{tikz/3d-tetrahedra--upper-bound-2--case7.tikz}%
    \else%
    \includegraphics{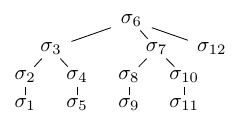}%
    \fi%
}%
        \label{fig:3d-tetrahedra:upper-bound-2:case7}
    }

    \subfigwidth=0.35\textwidth
    \captionsetup[subfigure]{width=\subfigwidth,justification=raggedright}
    \tikzmaybeexternalizenext{3d-tetrahedra--upper-bound-2--case8}
    \subfloat[]{
        \makebox[\subfigwidth][c]{%
    \ifbuildtikz%
    \input{tikz/3d-tetrahedra--upper-bound-2--case8.tikz}%
    \else%
    \includegraphics{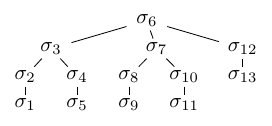}%
    \fi%
}%
        \label{fig:3d-tetrahedra:upper-bound-2:case8}
    }
    \hfill
    \subfigwidth=0.4\textwidth
    \captionsetup[subfigure]{width=\subfigwidth,justification=raggedright}
    \tikzmaybeexternalizenext{3d-tetrahedra--upper-bound-2--case9}
    \subfloat[]{
        \makebox[\subfigwidth][c]{%
    \ifbuildtikz%
    \input{tikz/3d-tetrahedra--upper-bound-2--case9.tikz}%
    \else%
    \includegraphics{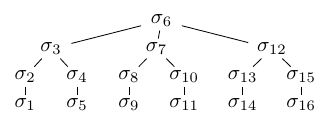}%
    \fi%
}%
        \label{fig:3d-tetrahedra:upper-bound-2:case9}
    }
    \tikzset{every picture/.append style={scale=1}}
    \caption{
The \enquote{nontrivial} configurations of children of \sig{6}.
The tree in~\ref{sub@fig:3d-tetrahedra:upper-bound-2:case1} is a 
subtree of all configurations.
In all cases, \sig{6} has no other children than those shown here.
Furthermore, since $T$ is rooted at a leaf node, \sig{6} 
needs to have an additional parent (except in 
case~\ref{sub@fig:3d-tetrahedra:upper-bound-2:case1}).
    }%
    \label{fig:3d-tetrahedra:upper-bound-2}
\end{figure}
\begin{proof}
We say that two beacons $b_1$ and 
$b_2$ \emph{share an edge} or 
are \emph{neighbors} if a 
a tetrahedron of 
$\Sigma$ contains an edge between the vertices
where $b_1$ and $b_2$ are placed.
We go through the cases.

\textbf{\cref{fig:3d-tetrahedra:upper-bound-2:case1}}: 
by \cref{lem:3d-tetrahedra:x-tetrahedra-share-y}\ref{lem:3d-tetrahedra:x-tetrahedra-share-y:four-share-vertex} 
the sets $\{\sig1, \sig2, \sig3, \sig6\}$ and 
$\{\sig6, \sig3, \sig4, \sig5\}$ each share one vertex, 
say $v_1$ and $v_2$, respectively.
If $v_1 \neq v_2$, we set
$B = \{v_1, v_2\}$.
If $v_1 = v_2$, we set $B = \{v_1, w\}$, where 
$w$ is any of the three other vertices of \sig{6}.
If \sig{6} has a parent tetrahedron, the shared facet 
contains three vertices of \sig{6} and hence at least one
beacon from $B$.
Thus, by placing $k = 2$ beacons,
we can remove the $6 = 3k$ tetrahedra $\{\sig{1}, \dots, \sig{6}\}$.

\textbf{\cref{fig:3d-tetrahedra:upper-bound-2:case2}}: 
we have the same situation as in
\cref{fig:3d-tetrahedra:upper-bound-2:case1},
except for the additional tetrahedron \sig{7}.
We choose $B$ as in 
\cref{fig:3d-tetrahedra:upper-bound-2:case1},
and we observe that \sig{6} contains two beacons.
Thus, \sig{7} contains at least one beacon from $B$.
Hence, by placing $k = 2$ beacons, we can remove 
the $7 > 3k$ tetrahedra
$\{\sig{1}, \dots, \sig{7}\}$.

\textbf{\cref{fig:3d-tetrahedra:upper-bound-2:case3}}: 
we apply the same argument as for 
\cref{fig:3d-tetrahedra:upper-bound-2:case2}, observing
that \sig{8} must also contain a beacon from $B$.
Thus, by placing $k = 2$ beacons, we can remove
the $8 > 3k$ tetrahedra
\sig{1} to \sig{8}.

\textbf{\cref{fig:3d-tetrahedra:upper-bound-2:case4}}: 
by 
\cref{lem:3d-tetrahedra:x-tetrahedra-share-y}\ref{lem:3d-tetrahedra:x-tetrahedra-share-y:three-share-edge},
the set $\{\sig6, \sig7, \sig8\}$ 
shares an edge $e$.
We apply 
\cref{lem:3d-tetrahedra:central-tetrahedron-1-subtree-with-5} 
to $\{\sig{1}, \dots, \sig{6}\}$. This gives two cases.
Case~(i): $\{\sig{1}, \dots,  \sig{5}\}$ share a vertex $v$.
As $v$ is in \sig{3}, and as \sig{3} shares a facet with \sig{6}, 
three neighboring vertices of $v$ are in \sig{6}.
The edge $e$ contains at least one of those three neighbors.
We call it $w$.
We set $B = \{v, w\}$.
Case~(ii): we obtain a vertex $v$ and an edge 
$e'$ in \sig{6}, with $v \cap e' = \emptyset$.
This covers three vertices of \sig{6}, so 
the edge $e$ shares at least one vertex with $v$ or
with $e'$.
To obtain $B$, we 
choose two vertices of \sig{6} such that $v$, $e'$, 
and $e$ each contain at least one.
In both cases, the beacons in $B$ are neighbors. 
We place $k = 2$ beacons,
and we remove the 
$7 > 3k$ tetrahedra
$\{\sig{1}, \dots,  \sig{5}, \sig{7}, \sig{8}\}$.

\textbf{\cref{fig:3d-tetrahedra:upper-bound-2:case5}}: 
by 
\cref{lem:3d-tetrahedra:x-tetrahedra-share-y}\ref{lem:3d-tetrahedra:x-tetrahedra-share-y:three-share-edge}, 
the sets $\{\sig6, \sig7, \sig8\}$ and 
$\{\sig6, \sig9, \sig{10}\}$ share edges 
$e_1$ and $e_2$, respectively.
We apply \cref{lem:3d-tetrahedra:central-tetrahedron-1-subtree-with-5} to 
$\{\sig{1}, \dots, \sig{6}\}$. This again gives two cases.
Case (i): $\{\sig{1}, \dots,  \sig{5}\}$ share a vertex $v$.
We set  $B = \{v, w_1, w_2\}$ 
such that $w_1$ and $w_2$ are 
vertices of \sig{6},
$|B| = 3$, and both edges $e_1$ and $e_2$ contain at 
least one beacon.
As in \cref{fig:3d-tetrahedra:upper-bound-2:case4}, $v$ 
is a neighbor of $w_1$ or $w_2$. Furthermore,  $w_1$ and $w_2$ 
are neighbors because they are vertices of \sig{6}.
Case (ii): we obtain a vertex $v$ and an edge 
$e$ in $\sig{6}$, with $v \cap e = \emptyset$.
We set $B = \{v, w_1, w_2\}$, where $w_1$ and $w_2$ are 
vertices of \sig{6}, such that $|B| = 3$ and such that all 
edges $e$, $e_1$, and $e_2$ contain at least one beacon.
Since all beacons lie in \sig{6}, they are mutual neighbors.
In both cases, we place $k = 3$ beacons such that 
every tetrahedron contains at least one.
We remove the $9 = 3k$ tetrahedra $\{\sig{1}, \dots, \sig{10}\} 
\setminus \{\sig{6}\}$.

\textbf{\cref{fig:3d-tetrahedra:upper-bound-2:case6}}: 
we apply 
\cref{lem:3d-tetrahedra:central-tetrahedron-1-subtree-with-5} 
to $\{\sig{1}, \dots, \sig{6}\}$ 
and to $\{\sig{6}, \dots,  \sig{11}\}$.
There are several cases.
Case~(i): each of 
$\{\sig{1},\dots,  \sig{5}\}$ and $\{\sig{7}, \dots, \sig{11}\}$ 
share a vertex, say $v_1$ and $v_2$, respectively.
By the argument from 
\cref{fig:3d-tetrahedra:upper-bound-2:case4},
three neighboring vertices of $v_1$ and three
neighboring vertices of $v_2$
are vertices of \sig{6}.
Thus, there is a vertex $v$ of \sig{6}
that is a neighbor of $v_1$ and of $v_2$.
We set $B = \{v, v_1, v_2\}$.
Case~(ii): without loss of generality, 
the set $\{\sig{1}, \dots,  \sig{5}\}$ shares
a vertex $v_1$
and the set $\{\sig{6}, \dots, \sig{11}\}$ 
has a vertex $v_2$ and an edge $e$ in 
\sig{6}, with $v_2\cap e = \emptyset$.
Then, at least one of the three vertices of \sig{6} 
that are neighbors of $v_1$ is covered by 
$v_2 \cup e$.
We set $B = \{v_1, v_2, w\}$, where
$w$ is an endpoint of $e$.
Case~(iii): $\{\sig{1}, \dots, \sig{6}\}$ 
have a vertex $v_1$ and an edge $e_1$ in \sig{6} and 
$\{\sig{6}, \dots, \sig{11}\}$ have a vertex $v_2$ 
and an edge $e_2$ in \sig{6}.
We choose for $B$ three vertices of \sig{6} 
such that $v_1$, $v_2$, $e_1$, and $e_2$ 
each contain at least one beacon.
In all cases, we place $k = 3$ beacons, so that
$B$ is connected and 
every tetrahedron in
$\{\sig{1}, \dots,  \sig{11}\}$ contains at least one beacon.
We remove $10 > 3k$ tetrahedra: all but \sig{6}.

\textbf{\cref{fig:3d-tetrahedra:upper-bound-2:case7}:} 
this is similar to \cref{fig:3d-tetrahedra:upper-bound-2:case6}.
We only describe how to ensure that $B$
contains a vertex of \sig{12}.
In Case~(i), $v$ can be placed at two vertices.
We choose the vertex that lies in \sig{12}.
This is always possible, as \sig{12} contains three 
of the four vertices of \sig{6}.
In Case~(ii), we choose $w$ as an
endpoint of $e$ that lies in \sig{12}.
The same argument as before applies.
In  Case~(iii), $B$ must contain a vertex of \sig{12},
since three beacons are at vertices of \sig{6}.
Thus, by placing $k = 3$ beacons, we remove 
$11 > 3k$ tetrahedra: all but \sig{6}.

\textbf{\cref{fig:3d-tetrahedra:upper-bound-2:case8}:} 
initially, we choose a set of beacons $B'$ as in
\cref{fig:3d-tetrahedra:upper-bound-2:case6},
at first
ignoring \sig{12} and \sig{13}.
By 
\cref{lem:3d-tetrahedra:x-tetrahedra-share-y}\ref{lem:3d-tetrahedra:x-tetrahedra-share-y:three-share-edge}, 
$\{\sig{6}, \sig{12}, \sig{13}\}$ shares an edge $e'$.
If $e'$ is covered by $B'$, we set $B = B'$. If not,
we set $B = B \cup \{w\}$, where $w$ is an
endpoint of $e'$. 
Thus, by placing $k \leq 4$ beacons, we may remove $12 \geq 3k$ 
tetrahedra: all but \sig{6}.

\textbf{\cref{fig:3d-tetrahedra:upper-bound-2:case9}:} 
let $S_1 = \{\sig1, \dots, \sig6\}$, 
$S_2 = \{\sig6, \dots, \sig{11}\}$, 
$S_3 = \{\sig6, \sig{12}, \dots, \sig{16}\}$. 
Also, let $S_1' = S_1 \setminus \{\sig6\}$, 
$S_2' = S_2 \setminus \{\sig6\}$, and 
$S_3' = S_3 \setminus \{\sig6\}$.
We apply \cref{lem:3d-tetrahedra:central-tetrahedron-1-subtree-with-5}
to $S_1$, to $S_2$, and $S_3$.
There are several cases.
Case~(i): $S_1'$, $S_2'$, and $S_3'$ 
each share a vertex, say $v_1$, $v_2$, and $v_3$.
By the argument of \cref{fig:3d-tetrahedra:upper-bound-2:case4},
each of $v_1$, $v_2$, $v_3$ has three neighbors that
are vertices of \sig{6}.
Thus, they have one common neighbor vertex $w$ in \sig{6}.
We set $B = \{v_1, v_2, v_3, w\}$.
Case~(ii): without loss of generality, 
$S_1'$ and $S_2'$ each share a common vertex,
say $v_1$ and $v_2$, and for $S_3$
we obtain a vertex $v_3$ and an edge $e_3$
in \sig{6}, with $e_3 \cap v_3 = \emptyset$.
We set $B =  \{v_1, v_2, v_3, w\}$, where $w$ is an
endpoint of $e$.
Since $v_3$ and $w$ are in \sig{6},
it follows that $v_1$ and $v_2$ have a neighboring beacon in \sig{6}.
Case~(iii): without loss of generality,
$S_1'$ has a common vertex $v_1$ and 
$S_2$ and $S_3$ each have a vertex $v_2$ and $v_3$ as well as
an edge $e_2$ and $e_3$, all four in \sig{6}.
We place a beacon at $v_1$ and three beacons at 
vertices of \sig{6} such that $v_2$, $v_3$, 
and both edges $e_1$ and $e_2$ contain at least one beacon.
Since three neighbors of $v_1$ are in
\sig{6}, the beacon at $v_1$ has at least one beacon 
neighbor in \sig{6}.
Case~(iv): $S_1$, $S_2$, and $S_3$ 
each have a vertex and an edge in \sig{6}.
We place three beacons so that all 
of them are covered.
In all cases, we place $k \leq 4$ beacons to 
remove $15 > 3k$ tetrahedra: all but \sig{6}.
\end{proof}

We are now ready to prove
\cref{lem:3d-tetrahedra:upper-bound}:

\begin{proof}[Proof (of {\cref{lem:3d-tetrahedra:upper-bound}})]
We use induction on the size of the
tetrahedral decomposition.
The base case is in \cref{lem:3d-tetrahedra:upper-bound:base-case}.
Next,
we assume that the inductive hypothesis 
(\cref{lem:3d-tetrahedra:upper-bound}) holds for all polyhedra 
that have a tetrahedral decomposition of size less than 
$m$.
Consider a spanning tree $T$ of the dual graph 
$D(\Sigma)$ of the tetrahedral decomposition $\Sigma$,
rooted at an arbitrary leaf.
Let \sig{1} be a deepest leaf. If 
\sig{1} is not unique, choose one with the largest number of 
siblings, breaking ties arbitrarily.
We can then apply either \cref{lem:3d-tetrahedra:upper-bound:inductive-step-i} 
or \cref{lem:3d-tetrahedra:upper-bound:inductive-step-ii},
to obtain the following:
\begin{proofcases}
\item\label{prf:3d-tetrahedra:upper-bound:beacons-tetrahedra}
we have placed a set $B$ of $k \geq 1$ beacons at
vertices of $\Sigma$, and we have removed 
at least $3k$ tetrahedra;
\item\label{prf:3d-tetrahedra:upper-bound:removed-tetrahedra-covered}
every removed tetrahedron contains at least one beacon in $B$;
\item\label{prf:3d-tetrahedra:upper-bound:beacons-neighbors}
the induced subgraph on $B$ on the vertices and edges of
$\Sigma$ is connected;
\item
\label{prf:3d-tetrahedra:upper-bound:beacon-in-remaining-polyhedron}
there is a beacon $b \in B$ in the remaining polyhedron $P'$.
\end{proofcases}

By~\ref{prf:3d-tetrahedra:upper-bound:beacons-tetrahedra},
the new polyhedron $P'$ has a tetrahedral decomposition of
size $m'\leq m - 3k < m$.
Thus, by the inductive hypothesis, 
we need 
\[
k' = \floor*{\frac{m' + 1}{3}} \leq 
\floor*{\frac{m - 3k + 1}{3}} =
\floor*{\frac{m + 1}{3}} - k
\]
beacons to route between any pair of points in $P'$.
Since $k' + k = \floor*{\textfrac{m+1}{3}}$, we 
do not exceed the claimed amount of beacons.
By the inductive hypothesis 
and~\ref{prf:3d-tetrahedra:upper-bound:beacon-in-remaining-polyhedron},
it follows in particular that we can route
from any point in $P'$ to the beacon $b \in B$ and vice versa. 
From~\ref{prf:3d-tetrahedra:upper-bound:removed-tetrahedra-covered},
we know that for every point $p$ in the removed tetrahedra, 
there is a beacon $b' \in B$ such that $p$ attracts $b'$ and 
$b'$ attracts $p$.
Finally, due to~\ref{prf:3d-tetrahedra:upper-bound:beacons-neighbors},
we can route between all beacons in $B$.
In conclusion, we can route between any pair of points in $P$.
This completes the inductive step.
\qedhere
\end{proof}

\begin{observation}%
\label{obs:coverage-is-bounded-by-routing}
\cref{lem:3d-tetrahedra:upper-bound}
also implies that 
$\max\{1, \floor*{\textfrac{m+1}{3}}\}$ beacons 
are sufficient to guard a polyhedron with a tetrahedral 
decomposition of size $m$.
We need at least one beacon to cover the polyhedron, and placing them 
as in the previous proof is enough.
\end{observation}


\section{A Lower Bound for Beacon\hyp{}based Routing}%
\label{sec:3d-tetrahedra:lower-bound}

Our next goal is to obtain a lower bound for the 
number of beacons needed to route in 
three-dimensional polyhedra.
We first give an alternative proof for the 
lower bound of 
$\floor*{\textfrac{n}{2}} - 1$ 
beacons for routing in two dimensions.
Our construction is similar to 
the one by Shermer~\cite{Shermer15} for 
orthogonal polygons.
We present a family of spiral\hyp{}shaped 
polygons for which we will then argue that 
$\floor*{\textfrac{n}{2}} - 1$ beacons are needed for 
routing between a specific pair of points.

\begin{definition}%
\label{def:spiral-polygon}
Given $c \in \N_{> 0}$ the \emph{$c$-corner spiral polygon}
is a simple polygon with $n = 2c + 2$ vertices
$s=r_0,r_1,\ldots,r_c,t=r_{c+1},q_c,q_{c-1},\ldots,q_1$,
in clockwise order.
The polar coordinates of the vertices are as follows: 
\begin{itemize}
    \item $r_k = \polar!{\floor!{\textfrac{k}{3}} + 1}
    {k\cdot \textfrac{2\pi}{3}}$, for $k = 0, \dots, c+1$; and
    \item 
     $q_k = \polar!{\floor!{\textfrac{k}{3}} + 1.5}
      {k \cdot \textfrac{2\pi}{3}}$, for $k = 1, \dots, c$.
\end{itemize}
The trapezoids $\trapezoid r_k q_k q_{k+1} r_{k+1}$,
for $k = 1, \dots, c - 1$ and the two triangles 
$\triangle s r_1 q_1$ and $\triangle t r_c q_c$ are 
called the \emph{hallways}.
\end{definition}

\begin{figure}[tbp]
    \centering
    \tikzmaybeexternalizenext{3d-tetrahedra--lower-bound--2d-class}
    \ifbuildtikz%
    \input{tikz/3d-tetrahedra--lower-bound--2d-class.tikz}%
    \else%
    \includegraphics{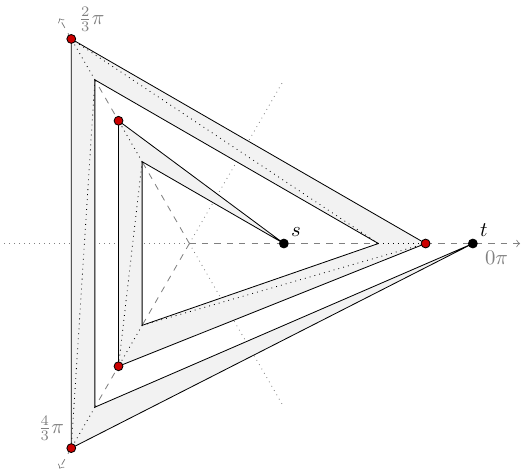}%
    \fi%

    \caption{
 A $5$-corner spiral polygon for which five beacons (marked in red)
 are necessary to route from $s$ to $t$.
    }%
    \label{fig:3d-tetrahedra:lower-bound:2d-class}
\end{figure}

An example for $c = 5$ is shown in 
\cref{fig:3d-tetrahedra:lower-bound:2d-class},
with a placement of five beacons
to route from $s$ to $t$.

\begin{figure}[tbp]
    \centering
    \subfigwidth=0.55\textwidth
    \captionsetup[subfigure]{width=\subfigwidth,justification=raggedright}
    \tikzmaybeexternalizenext{3d-tetrahedra--lower-bound--2d-class--details}
    \subfloat[Notation for the triangular spiral.]{
        \makebox[\subfigwidth][c]{%
    \ifbuildtikz%
    \input{tikz/3d-tetrahedra--lower-bound--2d-class--details.tikz}%
    \else%
    \includegraphics{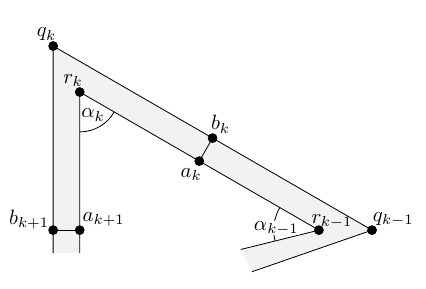}%
    \fi%
}%
        \label{fig:3d-tetrahedra:lower-bound:2d-class:details}
    }
    \hfill
    \subfigwidth=0.36\textwidth
    \captionsetup[subfigure]{width=\subfigwidth,justification=raggedright}
    \tikzmaybeexternalizenext{3d-tetrahedra--lower-bound--2d-class--corner}
    \subfloat[The complete corner $C_k$.]{
        \makebox[\subfigwidth][c]{%
    \ifbuildtikz%
    \input{tikz/3d-tetrahedra--lower-bound--2d-class--corner.tikz}%
    \else%
    \includegraphics{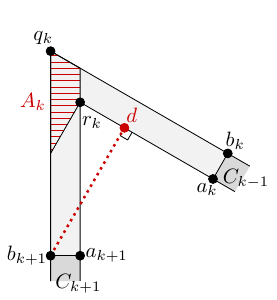}%
    \fi%
}%
        \label{fig:3d-tetrahedra:lower-bound:2d-class:corner}
    }
    \caption{
        A more detailed look at the parts of the spiral polygon.
    }
\end{figure}

\begin{lemma}[Two\hyp{}dimensional lower bound]
\label{lem:3d-tetrahedra:lower-bound-2d}
Let $c \in \N_{> 0}$ and let  $P$ be a 
$c$-corner spiral polygon. Let $B \subset P$ be a set
of beacons that lets us route from $s$ to $t$.
Then, we have $|B| \geq c$.
\end{lemma}

\begin{proof}
We shoot three rays from the origin with angles 
$\textfrac{\pi}{3}\pi$, $\pi$, and $\textfrac{5\pi}{3}$; 
see \cref{fig:3d-tetrahedra:lower-bound:2d-class}.
Each edge of $P$ is intersected by exactly one ray. 
For $k = 1, \dots, c + 1$, 
the intersection of a ray with the edge $r_{k-1}r_k$ 
is called $a_k$ 
and the intersection with the edge $q_{k-1}q_k$
is called $b_k$.
We divide $P$ into $c + 2$ subpolygons 
$C_0, \dots, C_{c+1}$ by drawing the line
segments $a_kb_k$, for $k = 1, \dots, c + 1$.
This gives two triangles $C_0$ and $C_{c + 1}$, 
with $s$ and $t$, respectively, and $c$ 
subpolygons $C_1, \dots, C_c$, called the
\emph{complete corners} of $P$; see 
\cref{fig:3d-tetrahedra:lower-bound:2d-class:details}.
We show that for $k = 1, \dots, c$, 
there must be at least one beacon from $B$
in $C_k \setminus (a_kb_k \cup a_{k+1}b_{k+1})$.

Suppose we route a point-shaped object $p$ from
$s$ to $t$ with the help of $B$.
Fix a complete corner 
$C_k$, $1 \leq k \leq c$, 
as in \cref{fig:3d-tetrahedra:lower-bound:2d-class:corner}.
Consider the last time the object $p$ crosses 
$a_kb_k$.
At this point, $p$ is attracted by
a beacon $b \in B$, and as we
require that $p$ moves all the way to $b$,
the beacon $b$ must lie in a complete corner
$C_\ell$, with $\ell \geq k$ (and $b$ is not on the
line segment $a_kb_k$).
In fact, $b$ can only be in 
$C_k$ or in $C_{k + 1}$, since otherwise it is clearly 
not possible that $p$ reaches $b$ along an 
attraction path.
Thus, for $p$ to reach $b$, it must be the case that 
either $a_kb_k$ is directly visible 
from $b$, or that the closest point to $b$ 
on $r_ka_k$ is $r_k$.
Otherwise, $p$ would get stuck on $r_ka_k$, 
see \cref{fig:3d-tetrahedra:lower-bound:2d-class:corner}.
The hatched region $A_k$ in
\cref{fig:3d-tetrahedra:lower-bound:2d-class:corner} shows 
the possible positions of $b$ under these constraints.
If this region is disjoint from $a_kb_k \cup a_{k+1}b_{k+1}$ the
claim follows immediately.

\begin{figure}[tbp]
    \centering
    \tikzmaybeexternalizenext{3d-tetrahedra--lower-bound--2d-class-hatched}
    \ifbuildtikz%
    \input{tikz/3d-tetrahedra--lower-bound--2d-class-hatched.tikz}%
    \else%
    \includegraphics{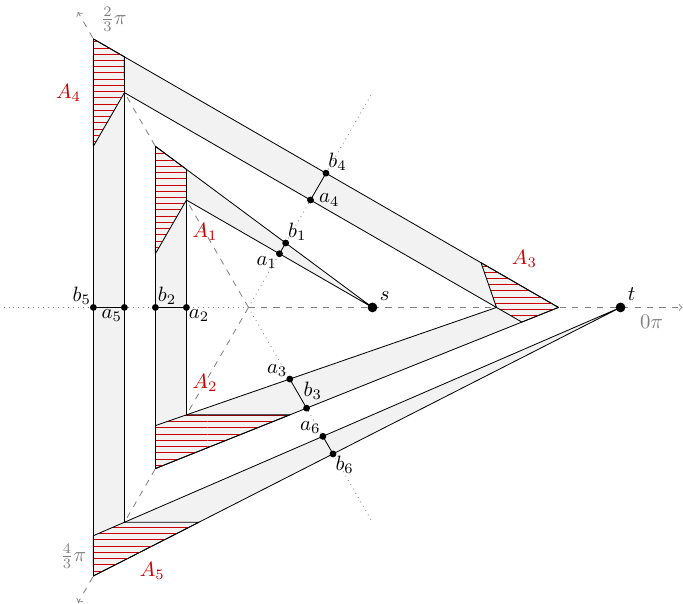}%
    \fi%

    \caption{
        A $5$-corner spiral polygon which shows the possible locations
        of the needed beacons to route through each corner when
        routing from $s$ to $t$.
    }%
    \label{fig:3d-tetrahedra:lower-bound:2d-class-hatched}
\end{figure}

In \cref{fig:3d-tetrahedra:lower-bound:2d-class-hatched} we can see all
$A_k$ for $1\leq k\leq c+1$ for $c=5$.
Clearly none of the $A_k$ intersect $a_kb_k$.
We show that none of the $A_k$ intersect $a_{k+1}b_{k+1}$ for each of the three directions:
\begin{enumerate}
\item
    $k=1,4,7,\ldots$: The $A_k$ are congruent since the angle $\alpha_k$ is always exactly $\pi/3$.
    Hence, as can be observed in \cref{fig:3d-tetrahedra:lower-bound:2d-class-hatched}, for increasing $k$ the distance from $A_k$ to $a_{k+1}b_{k+1}$ increases.
    Since $A_1$ does not intersect $a_2b_2$ the same holds true for all $k=1,4,7,\ldots$.
\item
    $k=2,5,8,\ldots$: The boundary edge of $A_k$ which could intersect $a_{k+1}b_{k+1}$ is always horizontal.
    As long as $b_{k+1}$ lies above this boundary edge no intersection is possible.
    This is the case for $A_2$ (as visible in \cref{fig:3d-tetrahedra:lower-bound:2d-class-hatched}).
    Since the length of the hallways increases and the angle $\alpha_k$ decreases for increasing $k$ it is always the case that $b_{k+1}$ lies above the horizontal bounding edge of $A_k$.
    Hence, none of the $A_k$ intersect $a_{k+1}b_{k+1}$ for $k=2,5,8,\ldots$.
\item
    $k=3,6,9,\ldots$: $A_3$ clearly does not intersect $a_4b_4$.
    However, as $k$ grows, the angle $\alpha_k$ increases towards $\pi/3$ and the $A_k$ grow towards a shape that is congruent with $A_1$.
    Since the hallways become larger and larger, even putting a rotated copy of $A_1$ at $A_3$ would not give an intersection with $a_4b_4$.
\end{enumerate}
It follows that $|B| \geq c$.
\end{proof}

We now extend this proof to three dimensions. For this,
we first define a $c$-corner spiral \emph{polyhedron}.

\begin{definition}%
\label{def:spiral-polyhedron}
Given $c \in \N_{>0}$ the \emph{$c$-corner spiral
polyhedron} is a 
polyhedron with $n = 3c + 2$ vertices $s=r_0$,
$r_1, \dots, r_c$, $t=r_{c+1}$, $q_1, \dots, q_c$,
and $z_1, \dots, z_c$.
The coordinates of $s$, $t$, $q_k$, and $r_k$, for 
$k = 1, \dots, c$,
are the same as in \cref{def:spiral-polygon}, 
with the $z$-coordinate set to $0$.
The $z_k$ are positioned above the 
corresponding $r_k$, i.e.,
$z_k = r_k + \brk*{\!\begin{smallmatrix}0\\0\\1\end{smallmatrix}\!}$,
for $k = 1, \dots, c$.
The edges and facets are given by the 
following tetrahedral decomposition:
\begin{itemize}
\item
The start and end tetrahedra are 
$\tetrahedron r_1 q_1 z_1 s$ and $\tetrahedron r_c q_c z_c t$.
\item
The \emph{hallway} between two triangles 
$\triangle r_k q_k z_k$ and  
$\triangle r_{k + 1} q_{k + 1} z_{k + 1}$ 
consists of the three tetrahedra 
$\tetrahedron r_k q_k z_k r_{k+1}$, 
$\tetrahedron r_{k+1} q_{k+1} z_{k+1} q_k$, 
and $\tetrahedron q_k z_k r_{k+1} z_{k+1}$, 
for $k = 1, \dots, c - 1$.
\end{itemize}
\end{definition}

For $c = 1$, the
$c$-corner spiral polyhedron has
two tetrahedra. For $c > 1$, we 
add $c - 1$ hallways, each 
with three tetrahedra. This 
means that a $c$-corner spiral polyhedron 
has a tetrahedral decomposition of size $m = 3c - 1$.
Thus, by \cref{def:spiral-polyhedron}, the number of tetrahedra 
in terms of the number of vertices is 
$m = 3 \cdot \textfrac{n - 2}{3} - 1 = n -3$, the smallest number
possible for a given $n$.

\begin{lemma}[Lower bound]%
\label{lem:3d-tetrahedra:lower-bound}
Let $c \in \N_{> 0}$ and let  $P$ be a 
$c$-corner spiral polyhedron. Let $B$ be a set
of beacons that lets us route from $s$ to $t$.
Then, $|B| \geq c$.
\end{lemma}
\begin{proof}
We show that a projection $B'$ of $B$ onto the $xy$-plane maintains the attraction regions.
It then follows from \cref{lem:3d-tetrahedra:lower-bound-2d} that $|B|=|B'|\geq c$.

Note that the only reflex edges in $P$ are the edges $e_k=r_kz_k$ for all $k=1,\ldots,c$.
Look at a beacon $b\in B$ and its projection $b'\in B'$.
If a point $p$ is visible from $b$ it must be visible from $b'$ as well:
Since the hallways are convex objects and the only edges that could prevent visibility are the vertical reflex edges $r_kz_k$ a vertical translation of $b$ to $b'$ cannot inhibit visibility.

If a point $p$ is attracted by $b$ (but not visible from $b$) it must be attracted by $b'$ as well.
Each such attraction goes through exactly one reflex edge: at least one since $p$ is not visible and at most one since two reflex edges in $P$ together form angles larger than $\pi$.
The movement of $p$ is a movement (possibly of length zero) until it hits a face $f_k=r_kz_kr_{k+1}z_{k+1}$ at point $q$.
It then slides along $f_k$ until it hits one of the boundary edges w.l.o.g. $e_k=r_kz_k$ at point $u$.
It then moves directly towards $b$.

Since $f_k$ is orthogonal to the $xy$-plane if $p$ is attracted by $b'$ it will hit $f_k$ at a point $q'$ which can be obtained by moving $q$ down along the $z$-axis.
Hence the point then slides from $q'$ along $f_k$ towards $e_k$ where it reaches at a point $u'$ which (again due to $e_k$ being orthogonal to the $xy$-plane) can be obtained by moving $u$ down along the $z$-axis.
It then moves directly towards $b'$.

Thus the set $B$ can only attract what $B'$ can.
Since $B'$ lies in the $xy$-plane and a cross section of $P$ along the $xy$-plane gives exactly a $c$-corner spiral \emph{polygon} $P'$.
By \cref{lem:3d-tetrahedra:lower-bound-2d} we obtain then that $|B| = |B'| \geq c$, as claimed.
\end{proof}


\section{A Tight Bound for Beacon\hyp{}based Routing}%
\label{sec:3d-tetrahedra:sharp-bound}

We combine the results from 
\cref{sec:3d-tetrahedra:upper-bound} and 
\cref{sec:3d-tetrahedra:lower-bound} into a tight bound:

\begin{theorem}\label{thm:3d-tetrahedra:sharp-bound}
Let $P$ be a three-dimensional polyhedron, and 
$m$ the smallest size of a 
tetrahedral decomposition of $P$.
Then, it is always sufficient and sometimes necessary to 
place $\floor*{\textfrac{m+1}{3}}$ beacons 
to route between any pair of points in $P$.
\end{theorem}

\begin{proof}
The upper bound was shown in  \cref{lem:3d-tetrahedra:upper-bound}.
For the lower bound, we consider the $c$-corner 
spiral polyhedron $P_c$ with 
$c = \floor*{\textfrac{m + 1}{3}}$. 
By \cref{def:spiral-polyhedron},
$P_c$ has a smallest tetrahedral decomposition of size 
$m' = 3c - 1$. 
Furthermore, by 
\cref{lem:3d-tetrahedra:lower-bound}, we need at least $c$ 
beacons to route in $P_c$. 
This also shows that $P_c$ does not have a
tetrahedral decomposition with size strictly less than $m'$,
since otherwise
\cref{lem:3d-tetrahedra:upper-bound} would 
yield a contradiction.

Due to the rounding we might have
$m' = m - 1$ or $m - 2$.
We then look at the $(c+1)$-corner spiral $P_{c+1}$ that consists
of three tetrahedra more than $P_c$.
More specifically, the last hallway of $P_{c+1}$ consists of the
three tetrahedra
$\sig{1}=\tetrahedron r_c q_c z_c r_{c+1}$,
$\sig{2}=\tetrahedron q_c z_c r_{c+1} z_{c+1}$, and
$\sig{3}=\tetrahedron q_c r_{c+1} q_{c+1} z_{c+1}$.
The tetrahedron \sig{1} is already present in $P_c$. Hence, 
for $m'=m-1$, we add
\sig{2}, and for $m'=m-2$, we add \sig{2} and \sig{3} to $P_c$.
Since for each additional tetrahedron we also need to add one
additional vertex ($z_{c+1}$ for \sig{2} and $q_{c+1}$ for \sig{3}),
there is no decomposition of the resulting polyhedron into less than
$m$ tetrahedra.

Additionally, the resulting polyhedron also
needs at least $c$ beacons because the added tetrahedra cannot lower
the number of beacons needed.
\end{proof}


\section{Conclusion}%
\label{sec:conclusion}

We have shown that, given a tetrahedral 
decomposition of a polyhedron $P$ 
of size $m$, we can place 
$\floor*{\textfrac{m + 1}{3}}$ beacons to route between 
any pair of points in $P$.
We also constructed a family of polyhedra 
where this is also necessary.

A lot of questions that have been studied 
in two dimensions remain open for the three-dimensional case.
For example, the complexity of 
finding an optimal beacon set to route between a given pair of 
points remains open.
Additional open questions concern the efficient computation of  attraction 
regions (computing the set of all points attracted 
by a single beacon) and of beacons kernels 
(all points at which a beacon can attract all points in the polyhedron).

Furthermore, Cleve~\cite{Cleve17} showed that not
all polyhedra can be covered by vertex beacons
and 
Aldana-Galv\'an~\etal~\cite{AldanaGalvanAlCaNeSoUrVe17a,AldanaGalvanAlCaNeSoUrVe17b} 
showed that this is even true for orthogonal polyhedra.
Given a polyhedron $P$ with a tetrahedral decomposition of 
size $m$, 
it remains open whether it is possible to guard 
$P$ with fewer than 
$\max\{1,\floor*{\textfrac{m + 1}{3}}\}$ beacons 
as in Observation~\ref{obs:coverage-is-bounded-by-routing}.

    \bibliography{bibliography}

    \appendix

\section{Program Code to Generate Trees}%
\label{appendix:trees-program}

\begin{lstlisting}
#!/usr/bin/env python3
"""Generate all dual graph configurations we need to look at."""

from itertools import product

from graphviz import Digraph


#############################################################################
# Tree structure.
#############################################################################
class Node:
    """A tree structure which allows pruning of unneeded subtrees."""

    # =======================================================================
    # General tree structure.
    # =======================================================================

    def __init__(self):
        """A new node is simply a leaf."""
        self.nodes = []

    def add(self, n=1):
        """Append n additional children and return self."""
        for _ in range(n):
            self.nodes.append(Node())
        return self

    def append(self, node):
        """Append a node or an iterable of nodes and return self."""
        try:
            for n in node:
                self.nodes.append(n)
        except TypeError:
            self.nodes.append(node)
        return self

    def is_leaf(self):
        """Return whether this node is a leaf, i.e., has no children."""
        return not self.nodes

    # =======================================================================
    # Graphviz.
    # =======================================================================

    def to_dot(self, graph=None, prefix=''):
        """Return a Graphviz representation of the tree."""
        if graph is None:
            graph = Digraph()
        self._dot_recursion(graph, 1, prefix)
        return graph

    def _dot_recursion(self, graph, current, prefix=''):
        """Recursively create Graphviz tree."""
        graph.node(prefix + str(current), label=str(current))
        this_number = current
        current = current + 1
        for child in self.nodes:
            current, child_number = child._dot_recursion(graph, current,
                                                         prefix)
            graph.edge(prefix + str(this_number), prefix + str(child_number))
        return current, this_number

    # =======================================================================
    # Pruning of "easy" cases.
    # =======================================================================

    def prune(self):
        """Remove subtrees that are easily removed."""
        # First try to remove subtrees.
        if self._prune() is None:
            return None

        # Call prune() for all children and filter out children that were.
        # removed
        self.nodes = list(filter(lambda x: x is not None,
                                 map(Node.prune, self.nodes)))

        # Sort children after pruning to have a canonical structure.
        self.nodes.sort()

        # Try pruning easy subtrees again. Maybe pruning the children created
        # a prunable configuration again.
        return self._prune()

    def _prune(self):
        """Remove subtrees that are easily removed."""
        if len(self.nodes) == 3:
            if all(n.is_leaf() for n in self.nodes):
                # Case (i): Figure 5.4(a): This is s2
                #   Three children that are leaf nodes: Remove all of them.
                self.nodes = []
            elif all(len(n.nodes) == 1 and n.nodes[0].is_leaf()
                     for n in self.nodes):
                # Case (iii)(3): Figure 5.4(e): This is s3
                #   Three children with one child leaf each: Remove two
                #                                            children.
                self.nodes.pop()
                self.nodes.pop()
        if len(self.nodes) == 2:
            if all(n.is_leaf() for n in self.nodes):
                # Case (ii): Figure 5.4(b): This is s2
                #   Two children that are leaf nodes: Remove both including
                #                                     the parent node.
                return None
        if len(self.nodes) == 1:
            if len(self.nodes[0].nodes) == 1:
                if self.nodes[0].nodes[0].is_leaf():
                    # Case (iii)(1): Figure 5.4(c): This is s3
                    #   A chain of three nodes: Remove all of them.
                    return None
        if len(self.nodes) >= 2:
            leaves = [n for n in self.nodes if n.is_leaf()]
            leaves2 = [n for n in self.nodes if len(n.nodes) == 1 and
                       n.nodes[0].is_leaf()]
            if leaves and leaves2:
                # Case (iii)(2): Figure 5.4(d): This is s3
                #   One leaf child and one child with a single leaf child:
                #       Remove both children.
                self.nodes.remove(leaves[0])
                self.nodes.remove(leaves2[0])

        # Return self to indicate that the node itself is not to be removed.
        return self

    # =======================================================================
    # Make trees comparable.
    # =======================================================================

    def __eq__(self, other):
        """
        Compare equality of two nodes.

        Two nodes are equal if they have the same number of children and
        every child is equal to the respective child of the other node.
        """
        if other is None:
            return False
        if len(other.nodes) != len(self.nodes):
            return False
        for this, that in zip(self.nodes, other.nodes):
            if this != that:
                return False
        return True

    def __lt__(self, other):
        """
        Compare whether a node is smaller than another node.

        A node is smaller then another node if it has more direct children or
        if any of the children is smaller than the respective other child.
        """
        if len(self.nodes) != len(other.nodes):
            return len(self.nodes) > len(other.nodes)

        for this, that in zip(self.nodes, other.nodes):
            if this < that:
                return True
            if that < this:
                return False

        return True

    # =======================================================================
    # String representation and hash value for uniqueness.
    # =======================================================================

    def __str__(self):
        """Generate a bracket term representing the tree."""
        return '(' + ''.join(str(n) for n in self.nodes) + ')'

    def __repr__(self):
        """Terminal representation."""
        return str(self)

    def __hash__(self):
        """Hash value for uniqueness."""
        return hash(str(self))


#############################################################################
# Generate all trees with certain maximum depth.
#############################################################################
def all_trees(depth):
    """
    Yield all trees with a given maximum depth.

    The trees are created recursively by appending combinations of trees of
    depth-1 to a node.
    """
    if depth == 1:
        # Create a node with 0, 1, 2, and 3 children.
        for i in range(4):
            yield Node().add(i)
    else:
        # Append 0, 1, 2, or 3 children.
        for number_of_children in range(4):
            # Create as many iterators of the next lower depth as there
            # should be children appended.
            next_level_iterators = []
            for _ in range(number_of_children):
                next_level_iterators.append(all_trees(depth - 1))

            # Combine all possible combinations of the iterators and add them
            # to a new node.
            for subtrees in product(*next_level_iterators):
                yield Node().append(subtrees)


def iterator_len(iterator):
    """
    Return the number of elements in an iterator.

    The iterator is consumed by calling this function.
    """
    length = 0
    for _ in iterator:
        length += 1
    return length


#############################################################################
# Main program.
#############################################################################
if __name__ == '__main__':
    # The maximum depth of the tree is 3
    depth = 3
    number_of_combinations = iterator_len(all_trees(depth))
    # Start with the first tree
    current = 1

    # A container for all distinct non-prunable trees
    trees = set()

    # Iterate through all different trees of maximum depth
    for tree in all_trees(depth):
        # Prune "easy" cases
        tree = tree.prune()
        # Add tree to set of trees if it was not pruned completely
        if tree is not None:
            trees.add(tree)

        # Debug output
        print('\r{percent:.2f}% ({current} / {all}) - trees: {trees}'
              .format(percent=100 * current / number_of_combinations,
                      current=current,
                      all=number_of_combinations,
                      trees=len(trees)),
              end='', flush=True)
        current += 1

    # Sum up the number of trees
    print()
    print(len(trees), 'trees')

    # Create a document with all non-prunable trees
    g = Digraph()
    prefix = 1
    for tree in trees:
        tree.to_dot(g, str(prefix) + '_')
        prefix += 1
    g.render('trees')
\end{lstlisting}

\end{document}